\documentclass[oneside,english]{amsart}
\usepackage[T1]{fontenc}
\usepackage[latin9]{inputenc}
\usepackage{amstext}
\usepackage{amsthm}
\usepackage{amsmath}
\usepackage{amssymb}
\usepackage{graphicx,color}
\usepackage{amsfonts,amscd}
\usepackage{mathtools}
\usepackage{paralist}
\usepackage{lipsum}
\usepackage{subfigure}
\usepackage{url}
\usepackage[font=footnotesize]{caption}
\usepackage{marginnote}
\usepackage{enumitem,comment}

\makeatletter

\numberwithin{equation}{section}
\numberwithin{figure}{section}
\theoremstyle{plain}
\newtheorem{thm}{\protect\theoremname}
\newtheorem{lemma}[thm]{Lemma}
\newtheorem{remark}[thm]{Remark}
\newtheorem{prop}[thm]{Proposition}

\newtheorem{corollary}[thm]{Corollary}
\newtheorem{criterion}[thm]{Criterion}

\newcommand{\ef}{\color{black}}

\newcommand{\Le}{L_\varepsilon}
\newcommand{\tLe}{{\tilde L}_\varepsilon}

\makeatother

\usepackage{babel}
\providecommand{\theoremname}{Theorem}

\usepackage[colorlinks=true]{hyperref}
\usepackage{marginnote}
\newcommand{\correction}[2]{#2}

\begin{document}
\title{On the Spectrum of Schr\"{o}dinger Operators Interacting at Two Distinct Scales}

\author{Emmanuel Fleurantin\textsuperscript{1}}
\author{Jeremy L. Marzuola\textsuperscript{2}}
\author{Christopher K.R.T. Jones\textsuperscript{3}}

\address[1]{Department of Mathematics, Morehouse College, Atlanta, GA 30314}
\address[2]{Department of Mathematics, University of North Carolina at Chapel Hill, Chapel Hill, NC 27599}
\address[3]{Department of Mathematical Sciences, George Mason University, Fairfax, VA 22030}

\email{emmanuel.fleurantin@morehouse.edu}
\email{marzuola@email.unc.edu} 
\email{ckrtj@renci.org}

\begin{abstract}
Schr\"{o}dinger operators of the form $\Delta - W$ on $L^2_{\text{rad}}(\mathbb{R}^3)$, the space of radially symmetric square integrable functions, are relevant in a variety of physical contexts. The potential $W$ is taken to be radially symmetric (i.e. $W(x) = W(|x|)$) and to decompose into two components with distinct spatial scales: $W=W_\varepsilon:= V_0+V_{1,\varepsilon}$. The second component $V_{1,\varepsilon}(|x|) = \varepsilon^2V_1(\varepsilon |x|)$ represents a scaled potential that becomes increasingly delocalized as $\varepsilon \to 0$. We will assume that both potentials $V_0(r), V_1(r)$ exhibit certain decay properties as $r \to \infty$. We show how the eigenvalue count on the positive real axis is built out of the spectra associated with the two reduced eigenvalue problems  on their separate scales. The result is that the total number of eigenvalues of $\Delta - W$ is the sum of the number of positive eigenvalues of $\Delta - V_0$ and $\Delta - V_1$. Our analysis combines dynamical systems techniques with a separation of scales argument, providing a novel framework for studying spectral properties of differential operators where multiple spatial scales interact.
\end{abstract}

\maketitle

\begin{flushleft}
\vspace{1em}

\end{flushleft}

\section{Introduction}

\subsection{Mathematical framework and setup}

The Schr\"odinger operator $\Delta - W$ on $L^2(\mathbb{R}^3)$ with certain conditions on the potential function $W$ arises in the classical study of quantum mechanics and in a variety of other physical contexts. In particular, the behavior of a quantum particle confined in a well described by the function $W$ is determined by the spectrum of the Schr\"{o}dinger operator $\Delta - W$, denoted $\sigma(-\Delta + W)$ , where the proper notions of discrete (or point) spectrum and essential spectrum are as described for instance in Reed-Simon~\cite{reed1978iv}, Chapter XIII and Kato~\cite{Kato1995}. As the problem is already sufficiently interesting, we work here in the setting of radially symmetric potentials.  I.e., we will restrict ourselves to the weighted space $L^2_{\text{rad}}(\mathbb{R}^3)$ consisting of radial functions $f(x) = f(|x|)$ with the norm
\begin{equation}
    \|f\|^2_{L^2_{\text{rad}}} = 4\pi\int_0^\infty |f(r)|^2r^2dr.
\end{equation}
Following Simon~\cite{Simon1982} and Agmon~\cite{Agmon1982}, we know the operator $\Delta - W$ is essentially self-adjoint on $C^\infty_0(\mathbb{R}^3)$, with its closure defining a self-adjoint operator with domain
\begin{equation}
    \mathcal{D} = \{f \in L^2_{\text{rad}}(\mathbb{R}^3): \Delta f - Wf \in L^2_{\text{rad}}(\mathbb{R}^3)\},
\end{equation}
where $\Delta f$ is understood in the distributional sense.
Given the orientation of our Schr\"odinger operator $\Delta - W$, we see that the essential spectrum, denoted $\sigma_{\rm ess} (\Delta -W)$, will consist of the negative real numbers, while the discrete spectrum, denoted $\sigma_p (\Delta - W)$, will consist of positive, isolated eigenvalues with associated eigenfunctions in $L^2_{\text{rad}}(\mathbb{R}^3)$.  Our focus here will be on quantifying the number of positive, eigenvalues for Schr\"odinger operators of a particular form.

We specifically consider the interesting (and challenging) case $\Delta - W$ on $L^2_{\text{rad}}(\mathbb{R}^d)$, the space of radial functions, where $W =W_\varepsilon= V_0 (|x|) + V_{1,\varepsilon} (|x|) $ is a sum of two radially symmetric potentials, sufficiently smooth and decaying,
which interact on different scales.  In particular, we take 
\[
V_{1,\varepsilon} (|x|) = \varepsilon^2 V_1 (\varepsilon |x|),
\]
to be a potential acting in the far field (adiabatically in space) with respect to the potential $V_0$.   For the radial problem, the operator takes the explicit form:
\begin{equation}
   L_\varepsilon := \Delta - W_\varepsilon = \frac{d^2}{dr^2} + \frac{2}{r}\frac{d}{dr} - V_0 - V_{1,\varepsilon}
\end{equation}
with eigenvalue equation
\begin{equation}
\label{eqn:eveq}
    u_{rr} + \frac{2}{r}u_r - \left(V_0  + V_{1,\varepsilon}\right)u = \lambda u.
\end{equation}

We will assume the following decay condition on the potentials. 

\vspace{.3cm}
\noindent{\bf (A1)}
{\em Both $V_0$ and $V_1$ are $C^\infty$ on $r\ge0$ and satisfy the decay conditions that there exists $C>0$ such that
    \begin{equation}
    \label{eqn:assump1}
    | \frac{d^k}{d r^k}V(r) | \leq \left(\frac{C}{r^{3+k+\gamma}}\right) \quad \text{for some } \gamma > 0, \ k=0,1,
    \end{equation}
for all $r \geq 1$, where $V$ is either $V_0$ or $V_1$.}

\correction{comment 7}{The assumptions on the potentials (A1) will be used to guarantee the validity (well-posedness) of the representation of the eigenvalue problem \eqref{eqn:eveq} as a dynamical system on a compact phase space. In particular, (A1) implies that the vector field on the compactified phase space is $C^1$.} Although, for this representation, decay like $1/r^{2+\gamma}$ would suffice, we also perform a blow-up procedure inside the compactified phase space and this requires an extra degree of decay. The decay conditions for $k=0$ also ensure that the interaction between the two potentials is not too strong as $\varepsilon \to 0$.  

We consider each of the operators $\Delta-V_0$, $\Delta-V_1$, and $\Delta-W_\varepsilon$ on the space: 
\begin{equation}\label{eqn:domain}
\mathcal{D} = \{u \in H^2(\mathbb{R}^3) \cap L^2_{\text{rad}}(\mathbb{R}^3): \Delta u \in L^2(\mathbb{R}^3)\}.
\end{equation}
The Lieb-Thirring inequalities~\cite{lieb2001inequalities}, and especially the Cwikel-Lieb-Rozenbljum (CLR) bounds, see for instance \cite{Cwikel1977,Lieb1976,Rozenbljum1976}, guarantee finiteness of the number of positive eigenvalues. For a potential $V=V(r)$, let us define 
\begin{equation}
\label{def:m}
    m(V) := \# \ \text{of positive eigenvalues of the Schr\"odinger operator $\Delta - V$} .
\end{equation}
The Cwikel-Lieb-Rozenbljum bound gives that there exists a constant ${\tilde C}_3$ independent of $V$ such that
\begin{equation}
    \label{eqn:clr}
    m(V) \leq {\tilde C}_3 \int_{\mathbb{R}^3} V_-^{\frac32} dx,
\end{equation}
where $V = V_+ - V_-$ with $V_+ \geq 0$ and $V_- > 0$ so $V_-$ is the component of the potential responsible for any trapping effects.  For these results, in addition to $V_- \in L^{\frac32}$, one typically assumes that $V_+ \in L^1_{\rm loc}$ to ensure that spectrum is well-defined, see \cite{hundertmark2023cwikel}.  It is clear that the $\frac32$ power on the right hand side of \eqref{eqn:clr} then means that 
\begin{equation}
\label{eqn:scaling}
\int_{\mathbb{R}^3} V_{1,\varepsilon,-}^{\frac32} dx = \int_{\mathbb{R}^3} V_{1,-}^{\frac32} dx.
\end{equation}
The exact value of the constant in these bounds is in and of itself a topic of great interest, see the recent result \cite{hundertmark2023cwikel}. 

It follows from assumption (A1) that both $V_0$ and $V_1$ are in $L^\frac32$ and $L^1_{loc}$. We can thus conclude that both spectral problems have only a finite number of positive eigenvalues. The scale separation in $W_{\varepsilon}$ introduces interesting spectral phenomena 
{\ef for self-adjoint operators, building on classical spectral theory for Schr\"{o}dinger operators, see \cite{reed1978iv, Simon1982, Kato1995}.} 
For eigenfunctions corresponding to discrete (point) spectrum, the eigenfunctions are "localized", in that they have exponential decay, however the rate will depend upon the distance to the essential spectrum.

A further assumption will be made about the end-point of the essential spectra of the two Schr\"odinger operators with potentials $V_0$, and $V_1$, namely that a resonance does not occur there. \vspace{.3cm}

\noindent {\bf (A2)}{\em The point $\mu=0$ is not an end-point resonance of either $\Delta -V_0$ or $\Delta -V_1$}
\makeatletter
\def\@currentlabel{A2}
\label{eqn:assump2}
\makeatother

\vspace{.3cm}

\correction{comment R2-1}{Note, the assumption here that there are no resonances is required for the proof. It is a stronger assumption than the absence of a kernel of the operators $\Delta -V_0$ or $\Delta -V_1$ in the space $\mathcal{D}$. For more details on the assumption, we refer to the paper by Sandstede and Scheel~\cite{sandstede2004evans}}. The point where we use this assumption involves a blow-up procedure, and is based on their work.

Schr\"odinger operators also arise in the study of linear stability theory in a variety of contexts.  In particular, Schr\"odinger operators of this type appear in reaction-diffusion systems in $\mathbb{R}^3$ with radially symmetric steady states. Consider the equation:
\begin{equation}\label{eqn:RDE}
    \frac{\partial u}{\partial t} = D\Delta u + f(u),
\end{equation}
where $D > 0$ is the diffusion coefficient and $f(u)$ represents reaction kinetics. When seeking steady states and linearizing around them, one obtains eigenvalue problems involving $\Delta - W$ where $W$ emerges from the linearized reaction terms.  

A particular motivation that brought us to this work is a scalar problem related to the study of nonlinear stability of solitons inside potential wells as considered in the papers of Nakanishi,  \cite{nakanishi2017global1,nakanishi2017global}.  When considering the nonlinear Schr\"odinger (NLS), which is the equation \eqref{eqn:RDE} with an $i$ in front of the time derivative term on the left hand side,  interacting with potentials at different scales, the operators are non-self-adjoint and consist of a more complicated matrix system of operators. Nevertheless, scalar spectral considerations turn out to be important to establish first, see for instance \cite{KriegerSchlag2006,NakanishiSchlag2011}.  The study presented here, while of independent interest and focusing on scalar Schr\"{o}dinger operators, grew out of the Ph.D. thesis \cite{golovanich2021potential}, which was directly motivated by the works of Nakanishi on NLS dynamics.

\subsection{Theorems}

There are three main results that together capture all the information about the spectrum of $\Delta-W_\varepsilon$, and show how it is constructed out of the separate spectra of $\Delta-V_0$ and $\Delta-V_1$. Note that it  follows from our assumptions that the essential spectrum is exactly $(-\infty,0]$ for all three operators: $\Delta-V_0$, $\Delta-V_1$, and $\Delta-W_\varepsilon$ (any $\varepsilon >0$).

\begin{thm}\label{thm:1} Under the assumptions (A1) and (A2), given $V_0, V_1$ there is an $\varepsilon_0 > 0$ such that for all $\varepsilon \in (0,\varepsilon_0)$, the operator $\Delta - W_{\varepsilon}$ on the domain $\mathcal{D}$ has:
\begin{enumerate}
\item No end-point resonance at $\lambda=0$, and
\item Spectrum in $\lambda>0$ consisting of exactly $m(V_0)+m(V_1)$ simple eigenvalues.
\end{enumerate}
\end{thm}

The next theorem shows that this structure, with its scale separation through $V_{1,\varepsilon}$, leads to rich spectral properties, particularly regarding eigenvalues in the gap between the essential spectrum $(-\infty,0]$ and the discrete spectrum associated with $\Delta - V_0$.  To the best of our knowledge, the precise description of the spectrum generated by two potentials separated by scales in this manner has not been described before in the literature.  It follows from the finiteness of the number of positive eigenvalues for the spectral problems, in particular for the potential $V_0$ that there is a spectral gap of the form $[0,\lambda_0]$, where $\lambda_0$ is the smallest positive eigenvalue of $\Delta -V_0$. 

\begin{thm}
    Under the same conditions as Theorem 1, with $\varepsilon$ sufficiently small, the operator $\Delta - W_{\varepsilon}$ has exactly $m(V_1)$ that are of $\mathcal{O}(\varepsilon^2)$.
\end{thm}

\begin{remark}
    \correction{comment R2-2}{A Rayleigh-Ritz argument in \cite{golovanich2021potential} actually shows that the eigenvalues inherited from the far field potential are all of at most $\mathcal{O}(\varepsilon ^2)$. The eigenvalues inherited from the near field potential only perturb a small amount when $\varepsilon$ is turned on.} 
\end{remark} 

The last result is stated as a corollary of these two theorems. It follows immediately, except for proving that the relevant eigenvalues are bounded away from $0$ as $\varepsilon\to 0$.

\begin{corollary}\label{cor}
     Under the same conditions as Theorem 1, with $\varepsilon$ sufficiently small, there are $m(V_0)$ eigenvalues of $\mathcal{O}(1)$.
\end{corollary}

Many other results exist to give integrability conditions on $W$ (and hence $V_0$, $V_1$) such that no discrete eigenvalues could exist, see \cite{bargmann1952number,berthier1982point}, and Theorem XIII.9 of \cite{reed1978iv} for an overview.  Interestingly, our proof of Corollary \ref{cor} shows that related conditions on $V_1$ make this potential still globally perturbative to $V_0$ in a natural way.

It is of note that the scalings we consider here that are sharp with respect to the (CLR) bound make the problem critical.  In particular, for a scaling of the form $\varepsilon^r V_1 (\varepsilon r)$, if $r> 2$ then the problem would be globally perturbative to $V_0$, while if $r <2$ it would be globally perturbative to $V_1$.

\subsection{Comparative analysis of dynamics on two distinct scales - Outline of the proof} At its heart, the proof uses a version of Sturm-Liouville Theory expressed through an angular flow. We transform the system so as to realize the eigenvalue problem as finding a heteroclinic connection in an appropriate phase space.  We then analyze two distinct model dynamical systems; one related to $V_0$ and the other to $V_1$. The key steps are as follows:

\begin{enumerate}
    \item \textbf{Domain compactification:} To study the behavior near spatial infinity and near the origin within a unified framework, we compactify the radial domain \( r \in [0, \infty) \) by introducing a new dependent variable related to $r$, and then transform it in its role as the independent variable. 
    \item \textbf{Separation into inner and outer problems:} We perform this compactification operation in two distinct ways  adapted to each scale. To focus on the near field, we use  $\sigma = \frac{r}{1+r} \in [0,1)$, while for the outer scaling regime we work with $\rho = \varepsilon r$, and in its compactified form \( \tau = \frac{\rho}{1+\rho} \in [0,1) \). These allow us to separate and analyze the distinct dynamics induced by the leading-order potentials \( V_0 \) and \( V_1 \) on their respective scales.
\item \textbf{Dynamics at end-point of essential spectrum:} For both the inner and outer problems, there is a singularity at this end-point which needs resolving. This is achieved by blow-up procedure carried out within the compactified systems.
    \item \textbf{Shadowing by base solutions:} We identify an intermediate region that lies between a region where the dynamics is given by the inner problem and one where the outer problem dominates. In each of the inner and outer regions, the solutions of the full problem can be well approximated by the respective model (or reduced) problems using standard methods of dynamical systems. To carry this shadowing into the intermediate region, where neither outer nor inner system rules, requires bootstrap estimates.

    \item \textbf{Spectral gluing:} In Section~\ref{sec:proof}, we use topological and spectral arguments to combine the eigenvalue counts obtained in the \( V_0 \) and \( V_1 \)-dominated regimes. This yields the full spectral information and completes the proofs of the three theorems given above.
\end{enumerate}

\subsection{Outline of the paper}
The paper is organized as follows. Section \ref{sec:geom} formulates the eigenvalue count through the two key scalings. We introduce the model systems and define a threshold for the spatial variable at which we match the model problems. Section \ref{endpoint} analyzes the endpoint $\mu=0$ of the essential spectrum, resolving the degenerate equilibrium that arises there via a spherical blow-up and establishing the continuity and decay-rate results needed to control the matching argument near the threshold. Section \ref{sec:estimates} develops the necessary machinery for tracking the relevant trajectories, determined by the asymptotic invariant manifolds, in the compactified phase space. Section \ref{sec:proof} completes the proof of Theorem \ref{thm:1} through a gluing argument at the threshold. Section \ref{sec:appl} provides numerical examples validating our theoretical results across different classes of potentials satisfying our hypotheses. Finally, Section \ref{sec:conclusion} discusses our findings and future directions.


\section{A geometric framework for the spectral problem}\label{sec:geom}
The eigenvalue problem on radial $L^2$ functions is a boundary value problem with a condition of regularity at $r=0$ and decay as $r \rightarrow \infty$. 
The eigenvalue equation is thus $\Le p=\lambda p$ with
\begin{equation}
\Le p = p_{rr} + \frac{2}{r} p_r - \left(V_0+V_{1,\varepsilon}\right) p =\lambda p
,\label{eq:Lp}
\end{equation}
where $V_0=V_0(r)$ and $V_{1,\varepsilon}(r)=\varepsilon^2V_1(\varepsilon r)$
The full eigenvalue problem for $\lambda >0$ being in the point spectrum of $L$
includes the boundary conditions
\begin{equation}
\lim_{r\rightarrow 0} p_r(r)=0, \hspace{2mm} {\rm and}\hspace{1mm}
\lim_{r\rightarrow \infty}p(r) =0
.\label{eq:LpBC}
\end{equation}
We will refer to this as the {\em inner system} as the potential in the far field is small, and it has a well-defined limit when $\varepsilon \to 0$ where only the inner potential $V_0$ remains. 

An analogous system will be derived  that focuses on the outer potential, called the {\em outer system}. This system will not readily have a limit as $\varepsilon\to 0$ without appropriate decay conditions on $V_0$. 

To formulate the outer system, we scale the independent variable $r=\frac\rho\varepsilon$, then
\begin{equation} \label{eq:Lp4}
\Le  p:=\varepsilon^2 p_{\rho\rho} +\varepsilon^2\frac2\rho p_\rho-\left( V_0\left(\frac\rho\varepsilon\right)+\varepsilon^2V_1\left(\rho\right)\right)p=\lambda p.
\end{equation}
Using \correction{comment R2-minor2}{$V_{0,\varepsilon^{-1}} \left( \rho \right)=\frac{1}{\varepsilon^2}V_0\left(\frac{\rho}{\varepsilon}\right)$} and defining $\mu$ such that 
\begin{equation}
\label{eqn:mudef}
    (\varepsilon \mu)^2 = \lambda,
\end{equation} which is justified as we will only need to consider $\lambda \ge 0$, we can rewrite \eqref{eq:Lp4} as
\begin{equation}\label{eq:Lp5}
    \tLe p :=p_{\rho\rho}+\frac2\rho p -\left(V_{0,\varepsilon^{-1}} \left(\rho\right) + V_1\left({\rho}\right)\right)p=\mu^2 p.
\end{equation}
\correction{comment R2-minor3}{Note that \eqref{eq:Lp5} remains $\varepsilon$-dependent through $V_{0,\varepsilon^{-1}}(\rho)$ and represents the eigenvalue problem expressed in outer coordinates, not a reduced limiting system.} The boundary conditions are, analogously,
\begin{equation}
\lim_{\rho\rightarrow 0} p_\rho(\rho)=0 \hspace{2mm} {\rm and}\hspace{1mm}
\lim_{\rho\rightarrow \infty}p(\rho) =0.
\label{eq:LpBCrho}
\end{equation}
If $\varepsilon > 0$, then\eqref{eq:Lp5}-\eqref{eq:LpBCrho} is equivalent to \eqref{eq:Lp}-\eqref{eq:LpBC}. Requiring a solution of either of these pairs of equations and boundary conditions to be in $L^2$ amounts to a decay condition as $r$ or $\rho$ tends to $+\infty$. With $\lambda>0$ solutions of \eqref{eq:Lp} either grow or decay exponentially, equivalently for $\mu>0$ in \eqref{eq:Lp5}. This justifies the second boundary condition in both \eqref{eq:LpBC} and \eqref{eq:LpBCrho}. The case of $\lambda =0$ (equivalently $\mu =0$) requires special consideration as all solutions decay, and they are distinguished by their rates of decay. This situation will be the subject of Section \ref{endpoint}.

\begin{remark}
    For the rest of the paper, we will assume that both $V_0$ and $V_1$ have been chosen and satisfy (A1) and (A2). All of the estimates will depend on these choices, in particular they depend on the parameters in the decay condition (A1). The key throughout will be keeping track of the dependence of the estimates on $\varepsilon$.
\end{remark}

\subsection{Eigenvalue problem as a dynamical system} Returning to the inner formulation of the eigenvalue problem \eqref{eq:Lp}, and
setting $p_r=v$, we can write \eqref{eq:Lp} as a first-order system
\begin{equation}
\begin{split}
  p'=&v, \\
  v'=& -\frac2r v +\lambda p + \left(V_0+V_{1,\varepsilon}\right) p,
  \end{split}\label{Lp1_sys}
\end{equation}
where ${}' = \frac{d}{dr}$.

The system \eqref{Lp1_sys} is a non-autonomous system on $\mathbb{R}^2$ with independent variable $r\in (0,\infty)$. A standard trick in dynamical systems is to make it autonomous by introducing $r$ as a new dependent variable. We will perform a variant of this strategy that has the effect of also compactifying the phase space. To that end, we set $\sigma=\frac{r}{r+1}$, and note that $\sigma'=\left( 1-\sigma \right) ^2$. 
The phase space of the resulting equation is
$\mathbb{R}^2\times (0,1)$, with the independent variable still being $r$.

We extend the potential functions by setting $\tilde{V}_0(1)=0$ and $\tilde{V}_{1,\varepsilon}(1)=0$.
The version of \eqref{Lp1_sys} augmented with $\sigma$ is
\begin{equation}\label{eq:Lp1-sys-comp}
\begin{split}
  p'=& v, \\
 v'=& 2\left(\frac{\sigma-1}{\sigma}\right) v +\lambda p + \left(\tilde{V}_0+\tilde{V}_{1,\varepsilon}\right) p,\\
  \sigma'=& (1-\sigma)^2,
  \end{split}
\end{equation}
which has as its phase space is $\mathbb{R}^2\times(0,1]$. In other words, we have compactified the phase space so that the limit $r\to+\infty$ now appears at $\sigma=1$.
\begin{remark}
    This compactification is an example of a general construction that is explored in some detail in \cite{wieczorek2021compactification}.
\end{remark}The vector field on the right-hand side of \eqref{eq:Lp1-sys-comp} is well defined on this phase space with the extended potentials. \correction{comment 3}{\correction{comment 1}{Its smoothness is addressed in the following lemma.
\begin{lemma}\label{lem:smooth-comp}
The right-hand side of \eqref{eq:Lp1-sys-comp} is $C^1$ on the phase space $\mathbb{R}^2\times(0,1]$.
    \end{lemma}
\begin{proof} The only issue is with the potentials at $\sigma=1$. We see that the difference quotient
$$\frac{1}{\sigma-1}\left(\tilde{V}_0(\sigma)-\tilde{V}_0(0)\right)=-\frac{1}{r+1}V_0(r)$$ 
has a limit $0$ by (A1) as $\sigma\to 1$, and so the derivative of the extended potential is $0$ at $\sigma=1$. From the chain rule, 
$$\frac{d}{d\sigma}\tilde{V}_o(\sigma)=\frac{1}{(1-\sigma)^2}V_0'(\frac{\sigma}{1-\sigma})=\frac{1}{(r+1)^2}V_0'(r),$$ and again the right-hand term tends to $0$ by (A1), showing the continuity of $\frac{d}{d\sigma}\tilde{V}_0(\sigma)$. The other potential is argued to be $C^1$ in a similar fashion.
    \end{proof} }}
\begin{remark} 
We note that the assumption (A1) is in fact stronger than needed for the proof of Lemma~\ref{lem:smooth-comp} and we could have proved the lemma with weaker decay. The stronger decay assumption is needed for the validity of the blow-up procedure in Section 3, see Lemma~\ref{lem:smooth-blowup}.
\end{remark}

The system \eqref{eq:Lp1-sys-comp} is singular at $\sigma =0$ since $\frac2r=2\left(\frac{1-\sigma}{\sigma}\right)$. To desingularize the system, we introduce a new independent variable 
\begin{equation}
    \label{eqn:sdef}
    s=r+\ln{r}.
\end{equation}  
\correction{comment R2-3}{Note that $\frac{ds}{dr} = 1 + \frac{1}{r} = 
\frac{1}{\sigma}$, so that the change of independent variable from $r$ to 
$s$ desingularizes the $\frac{1}{\sigma}$ factor appearing in the vector 
field near $r=0$.}
In these variables, the augmented version of\eqref{Lp1_sys} becomes
\begin{equation}
\begin{split}
  \dot{p}=&\sigma v, \\
 \dot{v}=& 2(\sigma-1) v +\sigma \left(\lambda p + \left(V_0+V_{1,\varepsilon}\right) p\right),\\
  \dot{\sigma}=&\sigma (1-\sigma)^2,
  \end{split}\label{Lp1_comp}
\end{equation}
where $\dot{}=\frac{d}{ds}$.
The eigenvalue condition now becomes one of finding a trajectory $\left( p(s),v(s),\sigma(s)\right)$ of \eqref{Lp1_comp} for which:
\begin{equation}\label{eq:BC1}\lim_{s\rightarrow -\infty}v(s)=0
\end{equation} and
\begin{equation}\label{eq:BC2}
\lim_{s\rightarrow +\infty} p(s)=0,
\end{equation} where we also require that $\sigma(s)\ne 0$ or $1$ for any finite $s$. 

This last condition requires some explanation which can also serve to illuminate this dynamical systems viewpoint.  The phase space for \eqref{Lp1_comp} is $\mathbb{R}^2 \times [0,1]$. With suitable extensions of the potentials to $\sigma =1$ by $0$, the planes $\sigma =0 $ and $\sigma =1$ are both invariant, with the former corresponding to $r=0$ and the latter the limit $r\rightarrow \infty$. The great benefit of this approach is that the asymptotic behavior, both at $0$ and $\infty$, is now represented within the phase space. The trajectory we want to get that corresponds to the eigenfunction lives in between these invariant planes. The condition that $\sigma \ne 0$ guarantees that we would not get a trajectory in the limit planes. 
To recover the eigenfunction, the parameterization needs to be set correctly. This is achieved by choosing $\sigma=\frac12$ when $s=1$. 

The analogous compactification and desingularization for the outer system involves setting 
\begin{equation}
    \label{eqn:tautdef}
\tau=\frac{\rho}{\rho+1}, \ \ t=\rho + \ln \rho.
\end{equation} 
It reads very similarly to \eqref{Lp1_comp} 
\correction{comment 2}{\begin{equation}
\begin{split}
  \dot{p}=&\tau w, \\
 \dot{w}=& 2(\tau-1) w +\tau \left(\mu^2 p + V_{0,\varepsilon^{-1}}+V_1\right) p,\\
  \dot{\tau}=&\tau (1-\tau)^2,
  \end{split}\label{eq:Lp5_comp}
\end{equation}}
where now $\dot{}=\frac{d}{dt}$. The boundary conditions \eqref{eq:BC1} and \eqref{eq:BC2} are the same with $s$ replaced by $t$ and $v$ by $w$.  The variable $p$ is the same in \eqref{Lp1_comp} and \eqref{eq:Lp5_comp} and $v=\varepsilon w$.

\subsection{The angular system}
Because the original eigenvalue equation is linear, system \eqref{Lp1_sys} is linear, as are the first two equations of \eqref{Lp1_comp} in $(p,v)$. We can then express $(p,v)$ in polar coordinates and the angular equations will decouple from the radial part. 

Setting the angular variable to be $\theta = \arctan(\frac{v}{p})$, we obtain the following system
\begin{equation}\label{eq:Lp3}
    \begin{split}
        \dot{\theta} = & (\sigma -1)\sin{2\theta} +\sigma \left(\left(\lambda + \tilde{V}_{0} (\sigma) + \tilde{V}_{1,\varepsilon} (\sigma) \right) \cos^2 \theta -\sin^2 \theta \right),\\\dot{\sigma} = & \sigma(1-\sigma)^2,
    \end{split}
\end{equation}
\normalsize
where $\dot{} = \frac{d}{ds}$. 
We can take $\theta$ to lie in either $\mathbb{R}$ or $S^1$.  The latter is achieved by taking $\theta \in [-\frac\pi2,\frac\pi2]$, mod $\pi$, which corresponds to taking a wrapped branch of $\arctan$.

  Another way to think of the angular system is for the angle of a subspace as it evolves under the linear dynamics of the first two equations of \eqref{Lp1_comp}. Since scalar multiples of solutions are also solutions, subspaces are carried to subspaces under the dynamics, which is why the angular equation decouples from the radial part.

On $\sigma =0$ the system reduces to $\dot{\theta}=-\sin 2 \theta$. On $\sigma =1$ the equation is $\dot{\theta} =\lambda \cos^2\theta -\sin^2 \theta$. The boundary conditions for the eigenvalue problem will be realized as asymptotic conditions to the relevant fixed points in these reduced equations.

Taking $\psi=\arctan(\frac{w}{p}) $, $\mu$ as in \eqref{eqn:mudef}, there is an analogous angular system for the outer equation:
\begin{equation}\label{eq:Lp6}
    \begin{split}
        \dot{\psi} = & (\tau -1)\sin{2\psi} +\tau \left(\left(\mu ^2+ \tilde{V}_{0,\varepsilon^{-1}} (\tau) + \tilde{V}_1 (\tau) \right) \cos^2 \psi -\sin^2 \psi \right),\\\dot{\tau} = & \tau(1-\tau)^2,
    \end{split}
\end{equation}
where $\dot{} = \frac{d}{dt}$ and $\tilde{V}_{0,\varepsilon^{-1}}$ and $\tilde{V}_1$ are extended in a similar fashion to $0$ as in \eqref{eq:Lp1-sys-comp}, recalling that $V_{0,\varepsilon^{-1}} \left( \rho \right)=\frac{1}{\varepsilon^2}V_0\left(\frac{\rho}{\varepsilon}\right)$.  To translate between $\theta$ and $\psi$, note that \begin{equation}
    \theta=\arctan \left( \frac{v}{p}\right)=\arctan \left( \frac {\varepsilon  w}{ p} \right),
\end{equation}
and so 
\begin{equation}\label{theta_psi}
    \tan \theta = \varepsilon \tan \psi.
\end{equation}

\subsection{Boundary conditions and invariant manifolds}
For \eqref{eq:Lp3} on $S^1\times [0,1]$, there are two equilibria in each of  $\sigma=0$ and $\sigma =1$. In $\sigma=0$, we have
\correction{comment 4}{\begin{equation}\label{eq:fps1}
     \hat{\theta}_{0,-} = \left(0,0\right), \hspace{3mm} \hat{\theta}_{0,+} = \left(\frac{\pi}{2}, 0\right),
\end{equation}
 and  in $\sigma =1$: 
\begin{equation}\label{eq:fps2}
\hat{\theta}_{1,\pm}(\lambda) = \left(\arctan (\mp\sqrt{\lambda}), 1\right).
\end{equation}}
The regularity boundary condition at $r=0$ will be satisfied by tending to $\hat{\theta}_{0,-}$, as that corresponds to $v=p_r=0$. For the boundary condition as $r \rightarrow +\infty$, it helps to see that the equation governing the dynamics in $\sigma =1$ is exactly the angular equation for the eigenvalue problem in one space dimension. The angle $\arctan \left( -\sqrt{\lambda}\right)$ is that of the subspace of decaying solutions, and the boundary condition is to find a trajectory tending to this value. 

Within each of their respective invariant planes, i.e., $\sigma =0 $ or $1$, the fixed points with a $+$ superscript are unstable (positive eigenvalue), while those with a $-$ superscript are stable (negative eigenvalue). Note the reversal of signs from the argument of the $\arctan$. Also, note the irony that the angle of the space of decaying solutions is actually unstable inside $\sigma =1$ as solutions that do not decay must grow and hence their angles are pushed away. 

Since $\sigma =0$ is repelling for \eqref{eq:Lp3} and $\sigma =1$ is attracting, albeit not exponentially, the fixed points $\hat{\theta}_{0,-} $ and $\hat{\theta}_{1,+} $ are both saddle-like, each with one attracting direction and one unstable direction. Although for $\hat{\theta}_{1,+}$ the attracting direction is a center direction (zero eigenvalue). 

We will, on occasion, include or suppress the dependence of $\theta$ on one or both of $\lambda$ and $\varepsilon$. The following proposition gives a characterization of eigenvalues in terms of the angular flow.
\begin{prop} \label{prop:eigcond1}
A value $\lambda >0 $ is an eigenvalue of $\Le$ for the space of radial $L^2$ functions iff \eqref{eq:Lp3} has a heteroclinic orbit $\left( \theta(s,\lambda,\varepsilon), \sigma (s) \right)$ connecting from  $\hat{\theta}_{0,-}$ to $\hat{\theta}_{1,+}$, i.e., 
\begin{equation}\label{eq:hetcondin}
\left( \theta(s,\lambda,  \varepsilon), \sigma (s) \right)\rightarrow \begin{cases} \hat{\theta}_{0,-} \hspace{3mm} {\rm as} \hspace{3mm} s \rightarrow -\infty,\\
     \hat{\theta}_{1,+}(\lambda) \hspace{3mm} {\rm as} \hspace{3mm} s \rightarrow +\infty. \end{cases}
     \end{equation}

\end{prop}

\begin{proof} 
A trajectory of \eqref{eq:Lp3}, for fixed $\lambda$ and $\varepsilon$, determines a curve, parameterized by $s$ of subspaces of solutions of \eqref{Lp1_sys}, and hence \eqref {eq:Lp}. These solutions are each regular at $r=0$ by the first condition of \eqref{eq:hetcondin}, and decay exponentially as $r \to + \infty$, from the second condition of \eqref{eq:hetcondin}. Hence $\lambda$ would be an eigenvalue of $\Le$. Conversely, if $\lambda >0$ is not an eigenvalue, then any solution in the subspace associated with a $\theta (s,\lambda,\varepsilon)$ that satisfies the first condition of \eqref{eq:hetcondin} will grow exponentially as $s \to +\infty$. In this case, $\left( \theta( s, \lambda,\varepsilon), \sigma (s) \right) \to \hat{\theta}_{1,-}(\lambda)$ as $s \to +\infty$.

\end{proof}

This eigenvalue condition can now be expressed in terms of invariant manifolds associated with these fixed points. We use the standard notation of $W^u$ as unstable, and $W^c$ center manifolds.

\correction{comment R2-4}{
\begin{lemma}\label{bc-im}
The boundary conditions \eqref{eq:BC1} and \eqref{eq:BC2} are thus satisfied if and only if 
\begin{equation}\label{eq:BC_IM1}
\left( \theta( s,\lambda,\varepsilon), \sigma (s) \right) \in W^u\left(\hat{\theta}_{0,-}\right),
\end{equation}
\begin{equation}\label{eq:BC_IM2}
\left( \theta( s,\lambda,\varepsilon), \sigma (s) \right)\in W^c\left(\hat{\theta}_{1,+}(\lambda)\right),
\end{equation}
respectively. Moroever, these are the unique trajectories in their respective invariant manifolds with $\sigma(s)\in [0,1]$.
\end{lemma}
 \begin{proof} 
 That conditions \eqref{eq:BC_IM1} and \eqref{eq:BC_IM2} correspond to \eqref{eq:hetcondin} is immediate. From the linearizations at $\hat{\theta}_{0,-}$ and $\hat{\theta}_{1,+}(\lambda)$, they are both $1-$dimensional, for fixed $\lambda$ and $\varepsilon$, and hence each consists of exactly one trajectory. The unstable manifold consists of all trajectories satisfying the boundary condition. While, in general, center manifolds are not uniquely defined, $W^c(\hat{\theta}_{1,+})$ is well-defined and unique as the complementary direction is (exponentially) unstable.  
\end{proof}
}
\begin{remark}
    The last point about uniqueness of the center manifold is a useful fact that will be used again later in related contexts. It is a standard fact in invariant manifold theory. The idea is that if the trajectories on the center manifold stay in the neighborhood in forward time, and there is no stable manifold present, then any other trajectory, i.e. not on the center manifold under consideration, must leave the neighborhood at some forward time. This comes from the complementary exponential separation in the unstable direction.
\end{remark}
Note that both $W^u(\hat{\theta}_{0,-})$  and $W^c(\hat{\theta}_{1,+})$ depend on $\varepsilon$ through \eqref{eq:Lp3}.
We will denote the trajectory in $W^u(\hat{\theta}_{0,-})$ by \begin{equation}\label{eq:thetahatminus}
    \hat{\theta}_-( s,\lambda,\varepsilon)=\left( \theta_-( s,\lambda,\varepsilon), \sigma (s) \right) ,
\end{equation}
and the one in $W^c(\hat{\theta}_{1,+})$ by
\begin{equation}\label{eq:thetahatplus}
    \hat{\theta}_+(s,\lambda,\varepsilon)=\left( \theta_+( s,\lambda,\varepsilon), \sigma (s) \right) .
\end{equation}
In both cases, we parameterize $\sigma(s)$ so that $\sigma\left(1\right)=\frac12 $, consistently with $\sigma = \frac{r}{r+1}$.

Both $\hat{\theta}_-(s,\lambda,\varepsilon)$ and $\hat{\theta}_+(s,\lambda,\varepsilon)$ are (at least) $C^1$ in $s \in \mathbb{R}$, $\lambda >0$, and $\varepsilon \in [0,\varepsilon_0)$, for some $\varepsilon _0 >0$. This follows from the unstable and center manifold theorems. 
We have only defined them for $\lambda >0$, but $\hat{\theta}_-(s,\lambda,\varepsilon)$ is actually defined for all $\lambda \in \mathbb{R}$. The trajectory determined as $r\to \infty$, namely $\hat{\theta}_+(s,\lambda,\varepsilon)$, is only defined for $\lambda >0$ A key step in the next section is to show it can be extended to $\lambda =0$.

A completely analogous development can be carried out for the outer system, represented by \eqref{eq:Lp6}. The fixed points in $\tau=0$ are the same as for the inner system. So we take $$\hat{\psi}_{0,-}=\hat{\theta}_{0,-} = \left(0,0\right),\hspace{3mm} \hat{\psi}_{0,+}=\hat{\theta}_{0,+} = \left(\frac{\pi}{2}, 0\right)$$ and 
$$\hat{\psi}_{1,\pm}(\mu) = \left(\arctan (\mp\mu), 1\right).$$ The analogue of Prop.~\ref{prop:eigcond1} then holds.

\begin{prop}
A value $\mu >0 $ is an eigenvalue of $\tilde{L}_\varepsilon$ for the space of radial $L^2$ functions iff \eqref{eq:Lp6} has a heteroclinic orbit $\left( \psi(t,\mu, \varepsilon), \tau (t) \right)$ connecting from  $\hat{\psi}_{0,-}$ to $\hat{\psi}_{1,+}$, i.e., 
\begin{equation}\label{eq:hetcondout}
\left( \psi(t, \mu, \varepsilon), \tau (t) \right)\rightarrow \begin{cases} \hat{\psi}_{0,-} \hspace{3mm} {\rm as} \hspace{3mm} t \rightarrow -\infty,\\
     \hat{\psi}_{1,+}(\mu) \hspace{3mm} {\rm as} \hspace{3mm} t \rightarrow +\infty. \end{cases}\end{equation}

\end{prop}

\begin{remark} Eigenvalues $\mu>0$ of $\tilde{L}_\varepsilon$ give eigenvalues of $\Le$ through multiplying by $\varepsilon ^2$, i.e., $\lambda =\varepsilon ^2\mu^2$ is an eigenvalue of $\Le$.
    
\end{remark}

The trajectories 
$$\hat{\psi}_-(t,\mu,\varepsilon)=\left( \psi_-(t,\mu,\varepsilon), \tau (t)\right), \ \ \hat{\psi}_+(t,\mu,\varepsilon)=\left( \psi_+(t,\mu,\varepsilon), \tau (t)\right)$$ 
are defined similarly as trajectories in $W^u(\hat{\psi}_{0,-})$ and $W^c(\hat{\psi}_{1,+})$ respectively, with appropriate parameterization in $t$.

\subsection{Matching condition and threshold}
The proof of Theorem~\ref{thm:1} involves showing that each potential provides a set of eigenvalues to the full problem. Those coming from $V_0$ are of order 1, and those associated with $V_1$ are order $\varepsilon ^2$. We use the inner and outer systems, respectively, to capture them as eigenvalues of $\Le$.

We will thus formulate two criteria for finding an eigenvalue. But note that satisfying either one renders an eigenvalue of $\Le$, using the relation $\lambda=\varepsilon^2\mu^2$. 

\begin{criterion}
    For fixed $\lambda >0$, if there is an $s\in \mathbb{R}$ so that
    \begin{equation}
    \theta_-(s,\lambda,\varepsilon)=\theta_+(s,\lambda,\varepsilon),\label{eq:cond1} \end{equation} 
   then $\lambda$ is an eigenvalue of $\Le$ with associated eigenfunction in $L^2$. 
\end{criterion}
Note that if \eqref{eq:cond1} holds for some $s$, it does so for all $s$ as both $\hat{\theta}_-$ and $\hat{\theta}_+$ are solutions of the same equation, \eqref{eq:Lp3}. The equivalent criterion, stated in terms of the outer system, is as follows.

\begin{criterion}
    For fixed $\mu >0$, if there is a $t\in \mathbb{R}$ so that
    \begin{equation}
    \psi_-(t,\mu,\varepsilon)=\psi_+(t,\mu,\varepsilon),\label{eq:cond2} \end{equation} 
   then $\mu$ is an eigenvalue of $\tilde{L}_\varepsilon$ with associated eigenfunction in $L^2$. 
\end{criterion}

 A key point will be to choose a value of $s$, respectively $t$, where the underlying trajectory hits a $\sigma$-section (respectively $\tau$-section) that lies between the inner and outer regions.  Since the potentials are originally functions of $r$, it is most natural to set an $r$ value first. To that end, we set $r_\varepsilon=\varepsilon^\alpha$ and so $\rho_\varepsilon=\varepsilon^{\alpha +1}$ where $\alpha\in \left(-1 , 0\right)$.

    \correction{comment R2-5}{The choice of $r_\varepsilon=\varepsilon^\alpha$ with $\alpha\in \left(-1 , 0\right)$ puts us into the region beyond where $V_0$ dominates, which is $O(1)$, corresponding to $\alpha =0$, and $O(1/\varepsilon)$ where $V_1$ dominates.} It will be seen later that we will actually require $\alpha < -\frac12$ to complete the proof.

The corresponding $s$ and $t$ will be 
\begin{equation}
\label{eqn:sepsdef}
    s_\varepsilon=\varepsilon^\alpha +\alpha \ln \varepsilon
\end{equation}
and
\begin{equation}
    \label{eqn:tepsdef}
    t_\varepsilon =\varepsilon^{\alpha +1} +\left( \alpha +1\right)\ln \varepsilon .
\end{equation}

The $\sigma$-section will then be \begin{equation}\label{sigma_threshold}
    \sigma_\varepsilon = \frac{\varepsilon^{\alpha}}{1+\varepsilon^{\alpha}},
\end{equation} 
and, in terms of $\tau$ \begin{equation}\label{tau_threshold}
   \tau_\varepsilon=\frac{\varepsilon^{\alpha+1}}{1+\varepsilon^{\alpha+1}}.
\end{equation}

\section{End-point of the essential spectrum}\label{endpoint} There is continuous spectrum on $(-\infty,0)$ and the end-point at $0$ requires special treatment. The implications of $0$ being the end-point are for the asymptotics of solutions of \eqref{eq:Lp}, although we will view this issue in terms of the angular systems \eqref{eq:Lp3} and \eqref{eq:Lp6}. 

There are three objectives in this section. First, we show that, in the outer system, the solution $\psi_-(t,\mu,\varepsilon)$ can be extended continuously to $\mu=0$ for fixed $t$ and $\varepsilon>0$. Secondly, we 
collect together the distinguished trajectories as $\mu$ varies and show this is a Lipschitz manifold. Thirdly, we
analyze the decay of the solutions of the inner system when both $\lambda =0$ and $\varepsilon =0$.   

The primary construction is a blow-up of the angular systems \eqref{eq:Lp3} and \eqref{eq:Lp6} at a degenerate fixed point in $\sigma=1$ ($\tau =1$) that appears when the eigenvalue parameter is at the endpoint of the essential spectrum.  

\subsection{A degenerate equilibrium at the endpoint}
Both angular systems become degenerate as the spectral parameter tends to zero, not because any individual term in \eqref{eq:Lp3} or \eqref{eq:Lp6} vanishes identically, but because the equilibrium $\hat\theta_{1,+}(\lambda)$ of \eqref{eq:fps2} loses hyperbolicity in the limit $\lambda\to0^+$. On the invariant plane $\sigma=1$ the reduced equation for \eqref{eq:Lp3} is $\dot\theta=\lambda\cos^2\theta-\sin^2\theta$, whose two roots are exactly $\hat\theta_{1,\pm}(\lambda)=(\arctan(\mp\sqrt\lambda),1)$. For every $\lambda>0$ these are distinct, hyperbolic equilibria of the reduced equation, since the linearization $-(\lambda+1)\sin(2\theta)$, evaluated at $\theta=\arctan(\mp\sqrt\lambda)\ne0$, is nonzero. As $\lambda\to0^+$, however, the two roots $\hat\theta_{1,+}(\lambda)$ and $\hat\theta_{1,-}(\lambda)$ collide at the common point $\theta=0$: at $\lambda=0$ the reduced equation is $\dot\theta=-\sin^2\theta$, whose zero at $\theta=0$ is a double root rather than a simple one, so the linearization there vanishes. It is this saddle-node-type collision of the two boundary equilibria, not a vanishing coefficient in the full vector field, that degenerates the equilibrium as $\lambda\to0^+$; the identical collision occurs on $\tau=1$ for \eqref{eq:Lp6} as $\mu\to0^+$. Combined with the center direction transverse to the invariant plane already noted in Section~\ref{sec:geom}, this is the loss of hyperbolicity anticipated there: the center manifold $W^c(\hat\theta_{1,+}(\lambda))$ used to define $\hat\theta_+(\lambda,s)$ in \eqref{eq:thetahatplus} does not persist, as such, all the way down to $\lambda=0$, and we must resolve the degenerate equilibrium directly.

We do so by a spherical blow-up. The construction is identical in the two angular systems, up to relabeling $(\theta,\sigma,\lambda)$ as $(\psi,\tau,\mu)$ and $\tilde V_0+\tilde V_{1,\varepsilon}$ as $\tilde V_{0,\varepsilon^{-1}}+\tilde V_1$. We carry it out in detail for the outer system \eqref{eq:Lp6} since we need the full 3-dimensional blow-up, where an equation for the eigenvalue parameter is appended, for the first objective mentioned in the preamble to this section. The corresponding blow-up of \eqref{eq:Lp3} at $\sigma=1$ follows verbatim, although for the second objective, which involves the inner system, we only need a 2-dimensional blow-up with $\lambda$ set equal to $0$.

Centering the boundary coordinate at $\tau=1$ via the shift
\begin{equation}\label{eq:s-shift}
    \zeta=1-\tau,
\end{equation}
so that $\zeta\geq 0$ near $\tau=1$, and appending $\dot\mu=0$ for the (eigenvalue) parameter, system \eqref{eq:Lp6} takes the form
\begin{equation}
\label{eq:local-normal-form}
\begin{aligned}
 \dot\psi
   &=-\zeta\sin(2\psi)
     +(1-\zeta)\bigl[W(\mu,\zeta)\cos^2\psi-\sin^2\psi\bigr],\\
 \dot \zeta&=\zeta^2(\zeta-1),\\
 \dot\mu&=0,
\end{aligned}
\end{equation}
where $W(\mu,\zeta):=\mu^2+\tilde V_{0,\varepsilon^{-1}}+\tilde V_1$, expressed as a function of $\zeta$ and $\mu$.
The origin $(0,0,0)$ is a fixed point of \eqref{eq:local-normal-form} and is fully degenerate, having $0$ as a triple eigenvalue. It is this point that we will blow up.

Since the potentials do not involve $\mu$, system \eqref{eq:local-normal-form}  is $C^1$ on $S^1\times [0,1] \times [0,\mu_0)$, where $\mu_0$ is any fixed, positive number. This follows from the same considerations, invoking assumption \eqref{eqn:assump1}, as given in the discussion just below \eqref{eq:Lp1-sys-comp}.

\subsection{Spherical blow-up of the degenerate equilibrium}

We blow up the origin in $(\psi,\zeta,\mu)$ by writing
\begin{equation}
\label{eq:blowup-map}
 \psi=R\widehat\psi,
 \qquad
 \zeta=R\widehat\zeta,
 \qquad
 \mu=R\widehat\mu,
 \qquad
 R = \sqrt{\psi^2+\zeta^2+\mu^2},
\end{equation}
where we require that $$
 (\widehat\psi,\widehat\zeta,\widehat\mu)\in\mathbb{S}^2.$$

Substituting \eqref{eq:blowup-map} into \eqref{eq:local-normal-form} and dividing the resulting vector field by the common factor $R$ (equivalently, rescaling the independent variable) yields a desingularized system having the following form
\begin{equation}
\label{eq:desingularized-system}
\begin{aligned}
 \widehat\psi'&=f(\widehat\psi,\widehat\zeta,\widehat\mu,R),\\
 \widehat\zeta'&=g(\widehat\psi,\widehat\zeta,\widehat\mu,R),\\
 \widehat\mu'&=h(\widehat\psi,\widehat\zeta,\widehat\mu,R),\\
 R'&=Rk(\widehat\psi,\widehat\zeta,\widehat\mu,R),
\end{aligned}
\end{equation}
where $R()'=\dot{}$ gives the appropriate reparameterization of the independent variable. 

The following lemma addresses the smoothness of the blown-up vector field on the right-hand side of \eqref{eq:desingularized-system}.
\begin{lemma}\label{lem:smooth-blowup} The vector field (right-hand side) of \eqref{eq:desingularized-system} is $C^1$ on $S^2\times [0,R_0)$ for some $R_0>0$
    \end{lemma}. 
    \begin{proof}
        Since \eqref{eq:desingularized-system} is obtained from \eqref{eq:local-normal-form}, we can see that the only component that needs to be checked is $f(\widehat\psi,\widehat\zeta,\widehat\mu,R)$. Within it, the terms that may cause problems have the form
        \begin{equation}\label{eq:scaled-potential-blowup}
            \frac{1}{R^2}V\left(\frac{1-R\widehat\zeta}{R\widehat\zeta}\right),\end{equation}
        where $V$ may come from $V_{0,\varepsilon^{-1}}$ (with $\varepsilon>0$ fixed) or $V_1$. 

        By considering the difference quotient and differentiating using the chain rule, \eqref{eq:scaled-potential-blowup} is seen to be $C^1$ on $S^2\times [0,R_0)$. The key point being that it is continuously differentiable at $R=0$. Note that the full force of (A1) is needed here due to the division by $R^2$.
    \end{proof}

Note that the hypersurface $R=0$ is invariant. For $R>0$, trajectories of \eqref{eq:desingularized-system} correspond, up to reparameterization, to trajectories of \eqref{eq:local-normal-form}, so the flow on $R=0$ records the limiting directions of orbits entering or leaving the degenerate equilibrium at the endpoint.

A full calculation of the vector field of \eqref{eq:desingularized-system} is unwieldy. We adopt the standard strategy of calculating the equation, and nearby, the sphere, using projective charts.

\subsection{Two local charts on the blown-up sphere}
Since $\zeta\ge0$ and $\mu\ge0$ throughout the original problem, only the quarter-sphere
\begin{equation*}
 \mathbb{S}^2_{++}
 =\bigl\{(\widehat\psi,\widehat\zeta,\widehat\mu)\in\mathbb{S}^2:
      \widehat\zeta\ge0,\ \widehat\mu\ge0\bigr\}
\end{equation*}
is relevant. Two projective charts cover $\mathbb{S}^2_{++}$, with the exception of the two points $(\pm 1,0,0)$, neither of which enter the analysis.

\subsubsection{The $\zeta$-chart}
We map the hemisphere where $\widehat\zeta>0$ to $\zeta=0$. The coordinates on this hemisphere (note $R=0$) are given by:
\begin{equation}
\label{eq:s-chart}
 \psi_1=\frac{\widehat\psi}{\widehat\zeta},
 \qquad
 \mu_1=\frac{\widehat\mu}{\widehat\zeta},
\end{equation}
and we can extend the map off the sphere as $\zeta=R\widehat\zeta$,

In order to calculate the equations on this chart, we observe that $\psi=\zeta\psi_1$ and $\mu=\zeta\mu_1$. These can be plugged directly into \eqref{eq:local-normal-form}. Working out only the needed lowest order terms, we obtain 
\begin{equation}
\label{eq:s-chart-desing}
\begin{aligned}
 \dot{\psi_1}&=\zeta\left( -\psi_1+\mu_1^2-\psi_1^2+\zeta (\psi_1^2-\psi_1)+\mathcal{O}(\zeta^2)\right),\\
 \dot{\zeta}\hspace{.8mm} &=\zeta^2(\zeta -1)\\
 \dot{\mu}_1&=\zeta \mu_1 (1-\zeta),
\end{aligned}
\end{equation}

We can rescale the independent variable in \eqref{eq:s-chart-desing} to obtain
\begin{equation}
\label{eq:s-chart-desing-rescal}
\begin{aligned}
 \psi_1'&= -\psi_1+\mu_1^2-\psi_1^2+\zeta (\psi_1^2-\psi_1)+\mathcal{O}(\zeta^2),\\
 \zeta'\hspace{.8mm} &=\zeta(\zeta -1)\\
 \mu'_1&=\mu_1 (1-\zeta),
\end{aligned}
\end{equation}
which has $\zeta =0$ as an invariant subset, which corresponds to the $\widehat{\zeta}>0$ part of the blown-up sphere expressed in the $\zeta$-chart.

The equations on $\zeta=0$ in this chart are
\begin{equation}
\label{eq:s-chart-sphere}
 \psi_1'=-\psi_1+\mu_1^2-\psi_1^2,
 \qquad
 \mu_1'=\mu_1.
\end{equation}
There are two equilibria of \eqref{eq:s-chart-sphere}: $(0,0)$ and 
 $(-1,0)$

 we denote the corresponding equilibria for \eqref{eq:s-chart-desing-rescal} by $Y_0=(0,0,0)$ and $Y_1=(-1,0,0)$.

The linearization of \eqref{eq:s-chart-desing-rescal} at $Y_1$ is
\begin{equation}
\label{eq:lin-Y1}
 \begin{pmatrix}
  1&2&0\\
  0&-1&0\\
  0&0&1
 \end{pmatrix},
\end{equation}
and so $Y_1$ is hyperbolic, with a one-dimensional stable manifold as well as a two-dimensional unstable manifold. 
Both unstable eigenvectors, $(1,0,0)$ and $(0,0,1)$ are tangent to $\{\zeta=0\}$, while the stable eigenvector $(-1,1,0)$ is the only one with a nonzero $\zeta$-component; transverse contraction occurs along this tilted direction rather than the coordinate $\zeta$-axis.

\subsubsection{The $\mu$-chart}
We map the hemisphere $\widehat\mu>0$ to $\mu =0$, and set
\begin{equation}
\label{eq:gamma-chart}
 \psi_2=\frac{\widehat\psi}{\widehat\mu},
 \qquad
 \zeta_2=\frac{\widehat\zeta}{\widehat\mu},
\end{equation}
where $\mu=R\widehat\mu$, and so $\psi=\mu\psi_2$ and $\zeta=\mu\zeta_2$. 
\begin{equation}
\label{eq:gamma-chart-desing}
\begin{aligned}
 \dot{\psi_2}&=\mu\left(-2\zeta_2\psi_2+1-\psi_2^2\right)+\mathcal{O}(\mu^2),\\
 \dot{\zeta_2}&=-\mu\zeta_2^2+\mathcal{O}(\mu^2)\\
\end{aligned}
\end{equation}
To extend away from $\mu=0$, we observe that $ \dot{\mu}=0$. Rescaling the independent variable, as before, we obtain the equation for the flow on this hemisphere, which is given by $\mu =0$
\begin{equation}
\label{eq:gamma-chart-sphere}
 \psi_2'=-2\zeta_2\psi_2+1-\psi_2^2,
 \qquad
 \zeta_2'=-\zeta_2^2.
\end{equation}
The equilibria are $(1,0)$, and $(-1,0)$. 

Denote the second equilibrium in the full $(\psi_2,\zeta_2,\mu)$ variables by $Y_2=(-1,0,0)$, the full
linearization in 3-dimensional space at $Y_2$ is
\begin{equation}
\label{eq:lin-gamma-chart}
 \begin{pmatrix}
 2&2&0\\
 0&0&0\\
 0&0&0
 \end{pmatrix}.
\end{equation}
The zero directions reflect the non-hyperbolic dynamics in the $\zeta$-direction tangent to the
blown-up sphere together with the fact that $\mu$ acts as a
parameter. Since $\zeta_2'=-\zeta_2^2$,
trajectories with $\zeta_2>0$ move toward the edge $\zeta_2=0$ in
forward time, but not exponentially.

\subsubsection{The flow on the quarter-sphere}

Together, the two charts cover
\[
    \mathbb{S}^2_{++}\setminus\{(\pm1,0,0)\},
\]
and overlap on the open region
\[
    \{\widehat\zeta>0,\widehat\mu>0\},
\]
which we will call the interior of the quarter-sphere. 

In spherical variables
$(\widehat\psi,\widehat\zeta,\widehat\mu,R)$,
\begin{equation*}
 Y_2=\left(
 -\frac1{\sqrt2},0
 \frac1{\sqrt2},
 0
 \right),
 \qquad
 Y_1=
 \left(
 -\frac1{\sqrt2},
 \frac1{\sqrt2},
 0,
 0
 \right).
\end{equation*}
The following proposition captures the part of the dynamics on the quarter-sphere.   
\begin{prop}\label{prop:het-sphere}
There is a heteroclinic connection in the quarter-sphere from $Y_1$ to $Y_2$. Moreover, it is the only trajectory in the quarter-sphere that tends to $Y_2$ in forward time.
    
\end{prop}
\begin{proof} There is a $1$-dimensional (local) center manifold, which we denote $W^c_{loc}$, inside the blown-up sphere at $Y_2$. From the second equation in \eqref{eq:gamma-chart-sphere}, the trajectory in $W^c_{loc}$ in the quarter-sphere tends to $Y_2$ in forward time. Moreover, it is the unique such trajectory as the complementary dynamics is exponentially repelling. The heteroclinic connection is found by following this trajectory in backward time. 

The trajectory can be followed by switching to the $\zeta$-chart and observing that $\mu_1 \to 0$ in backward time from the $\mu$ equation in \eqref{eq:s-chart-sphere}. Its limit set can only be $Y_0$ or $Y_1$. But the $\psi_2$ component of the trajectory is always bounded away from $0$ as $\psi_2'=\mu_1^2>0$ on $\psi=0$ from \eqref{eq:s-chart-sphere}, recalling that we are working here in backward time. Thus its limit set cannot be $Y_0$, and hence must be $Y_1$. This proves the proposition, and the heteroclinic connection is depicted in Fig. 3.1.
    
\end{proof}
This heteroclinic connection will act as a carrier through a neighborhood of the singular point $(0,0,0)$ in \eqref{eq:local-normal-form}.
\begin{figure}[htbp]
\centering
\includegraphics[width=0.9\textwidth]{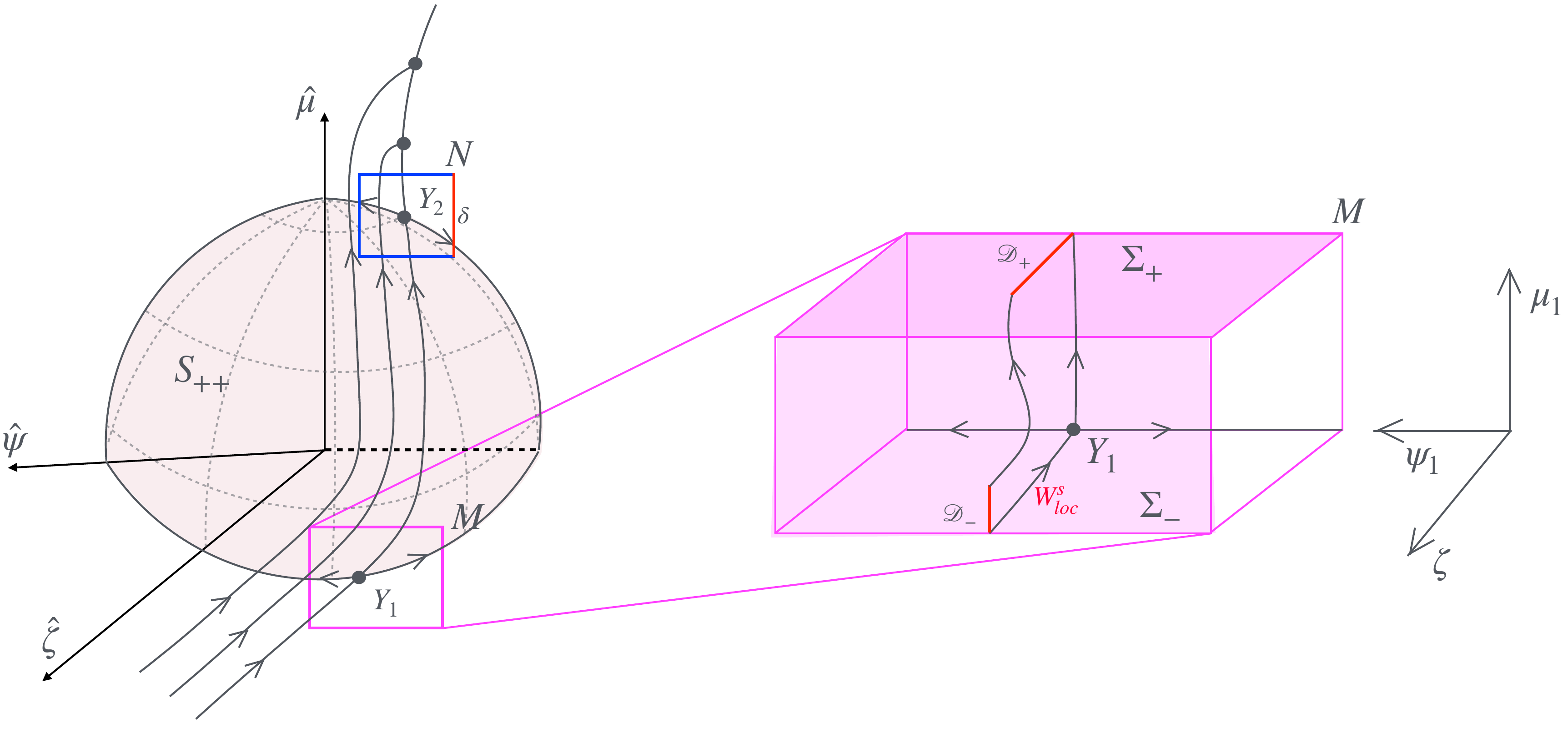}
\caption{The quarter-sphere $\mathbb{S}^2_{++}$ (left) and the local
neighborhood $N$ of $Y_1$ used in the continuity argument (right,
zoomed). In the zoomed box, the
transverse sections $\Sigma_+$ and $\Sigma_-$ and the incoming and
outgoing arcs $\mathcal{D}_+$ and $\mathcal{D}_-$ correspond to the construction
in Lemma~\ref{lem:contorbit}.}
\label{fig:quarter-sphere-flow}
\end{figure}

\subsection{Asymptotic dynamics with eigenvalue parameter at end-point of essential spectrum}
For $\mu>0$, we have a distinguished orbit $\psi_+(\mu,t,\varepsilon)$ for $\varepsilon$ sufficiently small. Recall that 
\begin{equation}\label{eq:out-dist-munot0}
    \psi_+(t,\mu,\varepsilon)\to \arctan (- \mu),
\end{equation}
as $t\to +\infty$.
So far, we have not defined a distinguished $\psi_+(0,t,\varepsilon)$ because $\arctan 0=0$ and there are many orbits of \eqref{eq:Lp6} with $\mu=0$ that converge to $0$ as $t\to +\infty$. 

If we view $Y_1$ in the $\zeta$-chart, we can see it has a $1$ dimensional (local) stable manifold, denoted $W^s_{loc}$. With  $\psi_1(t,\varepsilon)$, reparameterized back to $t$, being the trajectory in $W^s_{loc}$, we set 

\begin{equation}\label{eq:out-dist-mu0}
    \psi_+(t,0,\varepsilon)=\zeta (t) \psi_1(t,\varepsilon).
\end{equation}

Although it involves considerable undoing of the transformations that led to $\psi_1$, it can be checked that the solutions of the original problem \eqref{eq:Lp5} with angles $\psi_+(0,t,\varepsilon)$ will decay like $\frac{1}{\rho^2}$ and hence be in $L^2$, at least on the complement of a disc. The angle captured by $\psi_+(t,0,\varepsilon)$ is thus the appropriate one for the boundary condition of an eigenfunction candidate as $\rho$, or $r$, goes to $+\infty$.

\subsection{Continuity of the distinguished orbit in the eigenvalue parameter}

This section is devoted to proving the following lemma, which will allow us to pin down the eigenvalue count at $\mu =0$. 
\begin{lemma} 
\label{lem:contorbit}
For any fixed $\varepsilon >0$ and $t\in \mathbb{R}$, $$\lim_{\mu \to 0}\psi_+(t,\mu,\varepsilon) = \psi_+(t,0,\varepsilon).$$
In other words, $\psi_+(t,\mu,\varepsilon)$ is continuous in $\mu$ at $\mu=0$.

\end{lemma}
\begin{proof} 
The proof involves switching between the two charts we have formulated for the flow on, and near, the quarter-sphere. 
First, fix a small cube-like neighborhood $N$ of $Y_2$, which is expressed in the local coordinates $(\psi_2,\zeta_2,\mu)$ of the $\mu$-chart. Call the side length of the cube $\delta>0$.

The asymptotic limits of the candidate eigenfunctions $\psi_+(t,\mu,\varepsilon)$ are realized in these coordinates as the points 
\begin{equation}\label{eq:fp-extended-shpere}
\left(\frac{\arctan(-\mu)}{\mu},0,\mu \right),
\end{equation} each of which has a 1-dimensional center manifold due to the quadratic vector field in $\zeta$. The points in \eqref{eq:fp-extended-shpere} tend to $(-1,0,0)$ as $\mu\to 0$. But this is exactly the point $Y_2$ on the sphere. Its center manifold intersected with $\mu=0$ is the heteroclinic connection, see Proposition~\ref{prop:het-sphere}.

The center manifold $W^c_{loc}$ of $Y_2$ is the union of the center manifolds, each with $\mu$ fixed, of the points given in  \eqref{eq:fp-extended-shpere}. Choosing the neighborhood on which the full $W^c_{loc}$ is based to contain $N$, by making $\delta$ smaller if necessary, we set 
$$\mathcal{C}_{0}=W^c_{loc} \cap \Bigl\{\zeta_2=\frac\delta2,0\le\mu\le\frac\delta2\Bigr\},$$
which is then a curve parameterized by $\mu$ on the $\zeta_2=\frac\delta2$ face of $N$. For $\mu>0$, it is made up of the points where the distinguished orbits $\psi_+(\mu,t,\varepsilon)$ intersect this face of $N$. When $\mu=0$, it is the point where the heteroclinic connection, in the blown-up sphere from $Y_1$ to $Y_2$, intersects this face of $N$. Denote this point in $\mathcal{C}_0$ where $\mu=0$ by $c_0$

The next step is to realize $\mathcal{C}_0$ in $\zeta$-chart coordinates. Call this continuous curve $\mathcal{D}_0$, with $c_0$ being renamed $d_0$ in these coordinates.  We transport it under the backward flow of \eqref{eq:s-chart-desing-rescal} to  a neighborhood of $Y_1$. 

Since $Y_1$ is hyperbolic, we can invoke the Hartman-Grobman Theorem, see \cite{perko2013differential}, Section 2.8 for instance, to transform coordinates in a neighborhood of $Y_1$, so that \eqref{eq:s-chart-desing-rescal} becomes exactly linear. Note that this transformation can only be guaranteed to be continuous, but that is sufficient for our purposes. With an abuse of notation, we shall still use $(\psi_1,\zeta,\mu_1)$ for these coordinates, noting that this also includes a shift of $Y_1$ to the origin. 

We set a neighborhood $M$ of $Y_1$ to be a cube in the $(\psi_1,\zeta,\mu_1)$ coordinates, again of side $\delta$, centered at $Y_1$, and take $\delta$ to be as small as needed so that the equation is linear inside $M$, and hence of the following form 
\begin{equation}
\label{eq:linear-HG}
 \psi_1'=\psi_1,
 \qquad
 \zeta '=-\zeta,
 \qquad
 \mu_1 '=\mu_1.
\end{equation}

We denote the top face of the neighborhood
$\Sigma_+=\{\mu_1=\frac\delta2\}\cap M,$ and the front face to be $\Sigma_-=\{\zeta=\frac\delta2,0\le\mu_1\le\frac\delta2\}\cap M,$ both of which are (local) sections of the flow. Using the (backward) flow of \eqref{eq:s-chart-desing-rescal} we can transport $\mathcal{D}_0$ to a curve in $\Sigma_+$, which we denote $\mathcal{D}_+$. 

We parameterize $\mathcal{D}_+$ as $\{(\psi_1(a),\zeta(a),\mu_1): 0\le a \le a_0\}$, where $a_0>0$ is fixed so that $\zeta(a_0)=\frac\delta2$, making $\delta$ smaller if needed. Note that $\zeta (0)=0$ as it will lie on the blown-up sphere, and hence on the heteroclinic connection.  Under the backward flow in $M$, which comes from \eqref{eq:linear-HG}, both $\psi$ and $\mu_1$ will decrease when starting at points on $\mathcal{D}_+$, while $\zeta$ will increase, as long as we do not take the point where $\zeta(0)=0$. Since $\zeta(a_0)<\frac\delta2$, every point $\xi$ in $\mathcal{D}_+\backslash
\{(\psi(0),\zeta(0),\mu_1)\}$ will exit at a time $T=T(\xi)\le 0$ when $\zeta (T)=\frac\delta2$.

Denoting the flow of \eqref{eq:linear-HG} by $\Phi^T$, and $\xi=(\phi_1(a),\zeta(a),\frac\delta2)$ we can calculate
\begin{equation}
\label{eq:linear-solutions}
 \Phi^{T(\xi)}(\xi)=\left(\psi(t)=\psi_1(a)e^T,
 \zeta(a)e^{T},
 \mu_1 e^T\right).
\end{equation}
The exit time is determined by
\[
    \zeta(a)e^{T}=\frac\delta2,
\]
and hence
\begin{equation}
\label{eq:exit-time}
    T(\xi)
    =
    \ln{\left(\frac{2\zeta(a)}{\delta}\right)}
\end{equation}
Substitution into \eqref{eq:linear-solutions} gives the exit
coordinates
\begin{equation}
\label{eq:exit-coordinates}
 \psi_1
 =
 \frac{2\zeta(a)}{\delta}\psi_1(a),
 \qquad
 \mu_1
 =
 \zeta(a).
\end{equation}
Since $\psi_1(a)$ remains bounded and $\zeta(a)\to0$ as $a\to 0$, it follows
that, at the exit of $M$, $\psi_1 \to 0$ and $\mu_1 \to 0$.

In these (Hartman-Grobman) coordinates in $M$, the stable manifold, $W^s_{loc}$ of $Y_1$ lies on $\psi_1=\mu_1=0$, but this is, inside $M$, the distinguished orbit at $\mu=0$, which we denoted $\psi_+(0,t,\varepsilon)$, back in the original coordinates for the outer problem.
\end{proof}

\subsection{Lipschitz continuity of manifold of distinguished orbits}\label{sec:lip-mfld}
In the previous section, we considered the variation of the distinguished orbit in $\mu$ for fixed $\varepsilon$. We now fix $\mu=0$ and let $\varepsilon$ vary. To that end, we augment the outer angular system with $\varepsilon$, to obtain
\begin{equation}
\label{eq:local-normal-epsilon}
\begin{aligned}
 \dot\psi
   &=-\zeta\sin(2\psi)
     +(1-\zeta)\bigl[U(\zeta,\varepsilon)\cos^2\psi-\sin^2\psi\bigr],\\
 \dot \zeta&=\zeta^2(\zeta-1),\\
 \dot\varepsilon&=0,
\end{aligned}
\end{equation}
where $U(\zeta,\varepsilon)=\frac{1}{\varepsilon^2}V_0\left(\frac{1-\zeta}{\varepsilon\zeta}\right)+V_1\left(\frac{1-\zeta}{\zeta}\right)$. From the decay assumptions (A1) on the potentials, we can see that \eqref{eq:local-normal-epsilon} is $C^1$ on the phase space $S^1\times (0,1]\times [0,\varepsilon_0)$ for some $\varepsilon_0>0$, with, as usual, $U(\zeta,0)=0$.

We perform a partial blow-up on \eqref{eq:local-normal-epsilon} at $(0,0,0)$ in the $\psi$ and $\zeta$ directions. In the $\zeta$-chart, we have $\psi=\zeta \psi_1$, with $\zeta$ and $\varepsilon$ remaining the same. Rescaling the independent variable as above, we obtain
\begin{equation}\label{eq:zetachart-epsilon}
    \begin{aligned}
 \psi_1'&= -\psi_1-\psi_1^2+\zeta (\psi_1^2-\psi_1)+\mathcal{O}(\zeta^2),\\
 \zeta'\hspace{.8mm} &=\zeta(\zeta -1)\\
 \varepsilon'&=0.
\end{aligned}
\end{equation}
The fixed point $Y_1=Y_1(\varepsilon)=(-1,0,\varepsilon)$ also appears in this system, but with $\varepsilon$ being the third variable, rather than $\mu$. The distinguished orbit $\psi_=(0,\zeta,\varepsilon)=\zeta \psi_1$ is determined by the condition $\psi_1 \to -1$. At each point $Y_1(\varepsilon)$, the linearization at each $Y_1(\varepsilon)$ has eigenvalues $1,0,$ and $-1$. For fixed $\varepsilon$, the stable manifold comprises the distinguished orbit at that $\varepsilon$. Put together as $\varepsilon$ varies, they form a center-stable manifold. Note that it is unique as it is positively invariant and the complement is exponentially repelling. 

At $\varepsilon =0$, this center-stable manifold is given by the graph of a function $\psi_1=h(\zeta,\varepsilon)$. Since the vector field is $C^1$, the function is then guaranteed to be Lipschitz. Noting also that $\zeta$ is a continuous function of $t=\rho+\ln\rho$, and $\psi_+=\zeta \psi_1$, we have proved the following lemma.
\begin{lemma}\label{lem:dist-lipschitz}
    For fixed $t$, large enough, and some $\varepsilon_0>0$, the distinguished orbits $\psi_+(0,t,\varepsilon)$ are Lipschitz in $\varepsilon$, for $\varepsilon \in [0,\varepsilon_0)$.
\end{lemma}

\subsection{Decay rate at the endpoint: a two-dimensional blow-up} There is one more piece that we need to obtain from a blow-up analysis. It is actually the simplest in that it only requires a 2-dimensional blow-up, and it involves the inner system, which is less singular since there is no term of order $\frac{1}{\varepsilon^2}$. 

The aim is to prove the following lemma about the decay of the distinguished orbit of the inner model system, where $\lambda=\varepsilon=1$.
\begin{lemma}\label{Lemma:theta00}
    If $0$ is not an end-point resonance or eigenvalue of $\Delta -V_0$, then  $$\theta_0(r) = O\left( \frac{1}{r^2} \right)$$
    as $r\to +\infty$
\end{lemma}
\begin{proof}
Recalling the inner model system~\eqref{eq:Lp3}, and setting $\xi=1-\sigma$, we obtain:
\begin{equation}
\label{eq:local-normal-form_inner}
\begin{aligned}
 \dot\theta
   &=-\xi\sin(2\theta)
     +(1-\xi)\bigl[Y(\lambda,\xi)\cos^2\theta-\sin^2\theta\bigr],\\
 \dot \xi&=\xi^2(\xi-1),
\end{aligned}
\end{equation}
similarly to the derivation of \eqref{eq:local-normal-form}, where $Y(\lambda, \xi)=\lambda +  \tilde{V}_{0} (\sigma) + \tilde{V}_{1,\varepsilon} (\sigma)$. To obtain the inner model problem, we set $\lambda=\varepsilon=0$ in \eqref{eq:local-normal-form_inner}. Proceeding as before, we blow up \eqref{eq:local-normal-form_inner} at $(0,0)$ and consider it in the $\xi$-chart, where $\theta =\xi\theta_1$ and we have rescaled the independent variable. In a fashion completely analogous to the derivation of \eqref{eq:zetachart-epsilon}, except without the equation for $\varepsilon$, we obtain
\begin{equation}\label{eq:xichart-inner-model}
    \begin{aligned}
 \theta_1'&= -\theta_1-\theta_1^2+\xi (\theta_1^2-\theta_1)+\mathcal{O}(\xi^2),\\
 \xi'\hspace{.8mm} &=\xi(\xi -1).
\end{aligned}
\end{equation}

   On the blown-up circle at $\xi=0$ in  \eqref{eq:xichart-inner-model}, there are fixed points at $\theta_1=0$  and $\theta_1=-1$. 
    \correction{comment 8}{Each fixed point corresponds to a different asymptotic behavior of the original $p$. In the $\xi$ chart, we have $\theta = \xi\theta_1$ with $\xi \geq 0$ and $\xi \sim 1/r$ as $r\to \infty$. At $\theta_1=0$, $\theta \to 0$, so $v/p = \tan\theta \to 0$, which gives $p \to \mathrm{const}$. At $\theta_1=-1$, $\theta = \xi\theta_1 \sim -1/r$, so $p'/p = \tan\theta \sim -1/r$, which gives $p \sim C/r$.}
    In the former case, $0$ is not an endpoint resonance of $\Delta -V_0$, whereas  $0$ is an endpoint resonance of $\Delta -V_0$ in the latter case.
    \correction{comment 9}{Since $0$ is not an endpoint  resonance of $\Delta -V_0$ is assumed in the statement of the lemma, $\theta_0$ is attracted to the point at $\theta_1=0$ rather than $\theta_1=-1$.} \correction{comment 10}{Since $\theta_0 = \xi\theta_1$ with $\xi\sim 1/r\to 0$ and $\theta_1\to 0$ as $r\to\infty$, we have $\theta_0\to 0$. Linearizing \eqref{eq:Lp3} for small $\theta_0$ and using $\sigma-1\sim -1/r$ and $\tilde{V}_0\to 0$ as $r\to\infty$ gives $\dot{\theta}_0 = -\frac{2}{r}\theta_0 + O(\theta_0^2)$. Multiplying by the integrating factor $r^2$ and integrating yields $\theta_0(r) = O(1/r^2)$ as $r\to \infty$.}
    \end{proof}

    \begin{remark}The phenomenon discussed in the proof of Lemma~\ref{Lemma:theta00} was discovered earlier by Sandstede and Scheel \cite{sandstede2004evans}, where it is discussed in the context of contact defects and extensions of the Evans Function. \end{remark}

\section{Estimates for eigenvalue parameters close to zero}\label{sec:estimates}
The hardest part of our analysis is detecting eigenvalues close to zero, and we develop the needed estimates for this case first. Set $\lambda =\varepsilon^2\mu^2$ with $\mu$ order one. We need to estimate $\theta_-(s,\lambda,\varepsilon)$ from $s\to-\infty$ up to the threshold $s=s_\varepsilon$.  
We also recall coordinate $\psi$ defined by $\cot \psi =\varepsilon \cot \theta$.  
Of interest is $\psi_+$ which converges to $\arctan (-\mu)$ as $t \to \infty$. We will estimate how this varies with $\varepsilon$ at the threshold $t=t_\varepsilon$.

Both $\theta_-$ and $\psi_+$ will be estimated by the solution of a model problem. For the inner problem, the outer potential $V_{1,\varepsilon}$ is omitted. Similarly, for the outer problem, the inner potential \correction{comment 5}{$V_{0,\varepsilon^{-1}}$} is omitted. 
As will be seen in the proofs below, there is a region around the threshold where the estimates are harder to establish and we call this the {\em plateau}. In dynamical systems language, we need to show that the orbits $\theta_-$ and $\psi_+$ shadow their respective model problem trajectories with an error of $\mathcal{O}(\varepsilon$. This is straightforward over an interval of finite time, but we need to push both to a point where they can be compared, and this is into the plateau. The shadowing is achieved using bootstrap estimates and our explicit knowledge about the decay of the model solutions. 

\subsection{Inner problem}
We first augment the inner problem \eqref{eq:Lp3} with $\varepsilon$
\begin{equation}\label{eq:Lp3aug}
\begin{split}
        \dot{\theta} = & (\sigma -1)\sin{2\theta} +\sigma \left(\left(\varepsilon ^2 \mu ^2 + \tilde{V}_{0} (\sigma) + \tilde{V}_{1,\varepsilon} (\sigma) \right) \cos^2 \theta -\sin^2 \theta \right),\\\dot{\sigma} = & \sigma(1-\sigma)^2,\\\dot{\varepsilon}=&0
    \end{split}
    \end{equation}
    Equation \eqref{eq:Lp3aug} is $C^1$ on the phase space $S^1\times [0,1]\times[0,\varepsilon_0)$, for any $\varepsilon_0>0$, from assumptions (A1). The origin is a fixed point of \eqref{eq:Lp3aug} with eigenvalues $-2,0$ and $1$. 
    The (local) center-unstable manifold $W^{cu}(0,0,0)$ collects together the distinguished trajectories
    $$\hat{\theta} _-(s,\mu,\varepsilon)=\left( \theta_-(s,\mu,\varepsilon),\sigma (s),\varepsilon \right),$$ for $\varepsilon$ small. We denote the trajectory in this manifold with $\varepsilon =0$ by 
    $$\hat{\theta}_0(s)=\hat{\theta} _-(s,0)=\left( \theta_0(s),\sigma (s),0 \right),$$
    which is the trajectory lying in the $1$-dimensional unstable manifold. Note that $\varepsilon =0$ is an invariant subspace . 
 The inner model problem is the equation obtained by restricting to $\varepsilon=0$ in \eqref{eq:Lp3aug}. This leaves only the inner potential and an equation independent of $\varepsilon$. Dropping the equation for $\varepsilon$, it is then
\begin{equation}\label{eq:Lp_inner}
    \begin{split}
        \dot{\theta} = & (\sigma -1)\sin{2\theta} +\sigma \left(\tilde{V}_{0} (\sigma)\cos^2 \theta -\sin^2 \theta \right),\\\dot{\sigma} = & \sigma(1-\sigma)^2,
    \end{split}
\end{equation}
which is satisfied by $(\theta_0(s),\sigma(s))$.

In this subsection, we will prove the following lemma. We recall that $\sigma(s_\varepsilon)=\sigma_\varepsilon$ is the threshold at which gluing takes place.
\begin{lemma} \label{main_inner_estimate}
There is a constant $C_0$ and an $\varepsilon_0 >0$ so that

\begin{equation}\label{eq:innest}\left |\theta_-(s_\varepsilon,\mu,\varepsilon) -\theta_0(s_\varepsilon )\right|\le C_0 \varepsilon
\end{equation}
when $\varepsilon \le \varepsilon_0$ and $\mu$ is bounded. 
\end{lemma}

\begin{proof}

We will show that an estimate like \eqref{eq:innest} holds at successive values of $s$, culminating at $s_\varepsilon$. In the following the eigenvalue parameter $\mu$ does not play a role, and therefore we drop any explicit dependence on it.  To that end, set $$\nu_-(s,\varepsilon)=\theta_-(s,\varepsilon) -\theta_0(s ),$$
and define the sections $$\Sigma_i=\{(\theta,\sigma_i,\varepsilon):0\le \varepsilon \le \varepsilon_i\},$$ where $i=0,1$, $s_i$ is determined by $\sigma_i =\sigma(s_i)$. The value of $\sigma_0$ will be small and $\sigma_1\in (0,1)$ will be fixed independently of $\varepsilon$. The first two steps of the proof involve estimating $\nu_-(s,\varepsilon)$ at $s=s_0$, and then $s_1$.
\newline {\bf Step 1} The (local) center-unstable manifold at $(0,0,0)$, $W^{cu}_{loc}$, for \eqref{eq:Lp3aug} is two-dimensional and can be viewed as the graph of a $C^1$ function $h^{cu}$ from $$\{ (\sigma,\varepsilon):0\le \sigma \le \sigma_0, 0\le \varepsilon \le \varepsilon_0\}$$ to $\theta \in \mathbb{R}$, for some sufficiently small $\sigma_0$ and $\varepsilon_0$. 
Since $W^{cu}_{loc}$ is negatively invariant and the complementary direction is (exponentially) attracting, it is unique, and hence the curves $\hat{\theta}_-(s,\varepsilon)$ must lie in it for large enough negative $s$. Moreover, $h^{cu}$ is $C^1$ and $\theta_-(s,\varepsilon)=h^{cu}(\sigma (s),\varepsilon)$. It then follows that 
$\left| \nu_-(s,\varepsilon)\right| \le C \varepsilon,$
for some constant $C>0$, $\varepsilon \le \varepsilon_0$ and large negative $s$. In particular it holds for $s=s_0$, where $\sigma(s_0)=\sigma_0$ and $s_0$ is chosen large enough negative.
\newline {\bf Step 2}
The flow of \eqref{eq:Lp3aug} is $C^1$ and hence, by the Implicit Function Theorem, there is a $C^1$ transition map from $\Sigma_0$ to $\Sigma_1$, taking $\varepsilon_0$ smaller if necessary. From applying this map, it follows that $\left|\nu_-(s_1,\varepsilon)\right|\le C\varepsilon$, for some constant $C>0$ and sufficiently small $\varepsilon$. 
 This will be true up to any $s_1$, provided that lies in $(0,1)$ and is independent of $\varepsilon$. A value of $r=r_1$ will be set in the next section and $s_1=r_1 +\ln r_1$. Note that the technique of this step cannot be used to push the estimate to $s_\varepsilon$ since $s_\varepsilon \to \infty$ as $\varepsilon \to 0$.
\newline {\bf Step 3}  We will estimate $\nu_-(s,\varepsilon)$ up to the threshold $s=s_\varepsilon$ using a bootstrapping approach and the decay of $\theta_0(s)$ given in Lemma \ref{Lemma:theta00}. 

    \correction{comment 11} Using various trigonometric identities, we can calculate

    \correction{comment 12}{
    \begin{equation} \label{eq:dev_inner}
    \dot{\nu}_-=\left[2(\sigma -1)\cos 2 \theta_0-\sigma( \tilde{V}_0(\sigma)+1)\sin 2 \theta_0\right]\nu_- +h(\nu_-,\varepsilon),
    \end{equation}}
    where $\dot{\nu}=\frac{d\nu}{ds}$, as usual, and $h(\nu,\varepsilon)=O(\nu^2) +O(\varepsilon^2)$. 
    
    Since the decay of $\theta_0$ is given in terms of $r$, rather than $s$, it is clearer expressing the next part involving bootstrapping estimates in terms of $r$. Set $\eta(r) =\nu_-(s(r),\varepsilon)$, recalling that $s(r)=r+\ln r$. In these variables,  \eqref{eq:dev_inner} becomes 
    \begin{equation}\label{eq:dev_inner2}
     \eta'=-\left[\frac2r\cos 2 \theta_0+( V_0(r)+1)\sin 2 \theta_0)\right]\eta +h(\eta,\varepsilon),
    \end{equation}
    where $\eta'=\frac{d\eta}{dr}$ and $h(\eta,\varepsilon)=O(\eta^2) +O(\varepsilon^2)$.   \correction{comment 13}{Recall that throughout we have assumed that  $0$ is not an endpoint resonance for $ \Delta -V_0$.}

    We set $r_1$ to be independent of $\varepsilon$. The region $[r_1,r_\varepsilon]$ is called the inner plateau and we need to push the estimate
    across this part of the plateau from $r_1$ to $r_\varepsilon$. This is achieved using a bootstrapping argument. The bootstrap assumption will be that $\left|\eta(r)\right|\le C \varepsilon$ for all $r\in[r_1,r_\varepsilon]$ where the constant $C$ is independent of $\varepsilon$ and will be determined below. 

    \noindent Set $g(r)=\frac2r\cos 2 \theta_0+( V_0(r)+1)\sin 2 \theta_0$, which is the coefficient of $\eta$ in \eqref{eq:dev_inner2}. Using an integrating factor
\begin{equation}\label{eq:int_fac}
\Gamma(r)=e^{ \int_{r_1}^rg(s)ds},
\end{equation}
we can write
\begin{equation} \label{eq:dev_inner_intfac}
    \Gamma^{-1}(r)\Gamma(r_1)\eta(r_1)+\Gamma^{-1}(r)\int_{r_1}^r\Gamma(z)h(\eta(z),\varepsilon)dz.
\end{equation}
The integrating factor $\Gamma(r)$ can be estimated as follows.  Adjusting our choice of $r_1 $ sufficiently large if necessary, but still independent of $\varepsilon$, we can find a small constant $0 < \delta < \frac14$ depending only on $r_1$ and $V_0$ such that
$$1-\frac\delta4<\cos 2\theta_0(r)<1+\frac\delta4,$$
and 
$$\left |\frac r2 \left(V_0(r)+1)\sin 2 \theta_0(r)\right)\right |<\frac\delta4 $$
hold for all $r\ge r_1$ from Lemma \ref{Lemma:theta00}. Note a key here is that Lemma~\ref{Lemma:theta00} does not involve $\varepsilon$ and so the setting of $r_1$  does not depend on $\varepsilon$. It follows that
\begin{equation}\label{eq:g_estimate}\frac{2-\delta}{r}<g(r)<\frac{2+\delta}{r}\end{equation}
for all $r\ge r_1$, and so
\begin{equation}\label{gamma_estimate}\left(\frac{r}{r_1}\right)^{2-\delta}<\Gamma(r)<\left(\frac{r}{r_1}\right)^{2+\delta}.\end{equation}
Plugging into \eqref{eq:dev_inner_intfac}, we estimate
\begin{equation}
    \left|\eta(r)\right|\le \left|\eta(r_1)\right|+c_1\left(\frac{r}{r_1}\right)^{\delta -2}\int_{r_1}^r \Gamma(z)dz\left(\eta^2 +\varepsilon^2\right),
\end{equation}
for some constant $c_1(r_1,\delta)$ independent of $\varepsilon$. Using the bootstrapping assumption, \eqref{gamma_estimate} again and integrating, we have
\begin{equation}\label{bstrap_est1}
    \left|\eta(r)\right|\le \left|\eta(r_1)\right|+ (C^2 + 1) c_1 r^{1+2\delta}\ \left(\varepsilon^2\right).
\end{equation}
The second term on the right hand side will be largest on the interval $[r_1,r_\varepsilon]$ at $r=r_\varepsilon=\varepsilon^\alpha$, and so 
\begin{equation}\label{bstrap_est2}
    \left|\eta(r)\right|\le \left|\eta(r_1)\right|+ (C^2 + 1) c_1\varepsilon^{\alpha(1+2\delta)}\  \varepsilon^2.
\end{equation}
By the bootstrapping hypothesis, \correction{comment 14}{we have}
\begin{equation}\label{bstrap_est3}
    \left|\eta(r)\right|\le \left|\eta(r_1)\right|+(C^2 + 1)  c_1 \varepsilon^{\alpha(1+2\delta)+2}.
\end{equation}
We choose $-1 < \alpha< 0$ so that $\alpha ( 1+2 \delta) > -1$, i.e
\begin{equation}
    2 +\alpha(1+2\delta)=1+\gamma,
\end{equation}
for some sufficiently small $\gamma >0$ so that $-1/(1+2 \delta) <\alpha<0$.

Recall the bootstrapping assumption that $\left|\eta(r)\right| \le C\varepsilon $ for all $r\in[r_1,r_\varepsilon]$ for $C$ to be determined, which we now fix using the behavior as $\eta$ as we enter the plateau region $r>r_1$.
We know from the invariant manifold estimates above that we can find $\varepsilon _0>0$ and a $C>0$ independent of $\varepsilon$ so that 
$\left | \eta(r_1)\right | \le \frac C4 \varepsilon$ for all $\varepsilon < \varepsilon_0$.  This gives us the bootstrapping constant we need.  Adjusting $\varepsilon_0 $ if necessary small enough that \correction{comment 15}
{$(C^2+1)c_1 \varepsilon^{\gamma}<\frac14$} for $\varepsilon<\varepsilon_0$, we can conclude that
\correction{comment 17}{
\begin{equation}
  \left|\eta(r)\right|\le \frac{C}{2}\varepsilon,  
\end{equation}}
for all $r\in[r_1,r_\varepsilon]$, which closes the bootstrapping argument establishing the bound as claimed in the lemma when transferred back to $s$.
\end{proof}

 \subsection{Outer problem}
 For the outer problem, augmenting by $\varepsilon$ leads to
\begin{equation}\label{eq:Lp3_outeraug}
\begin{split}
        \dot{\psi} = & (\tau -1)\sin{2\psi} +\tau \left(\left(\mu ^2 + \tilde{V}_{0,\varepsilon^{-1}} (\tau) + \tilde{V}_{1} (\tau) \right) \cos^2 \psi -\sin^2 \psi \right),\\\dot{\tau} = & \tau(1-\tau)^2,\\\dot{\varepsilon}=&0
    \end{split}
    \end{equation}
    
    This is $C^1$ on the phase space $S^1\times (0,1]\times[0,\varepsilon_0)$ from assumptions (A1). Note that $\tilde{V}_{0,\varepsilon^{-1}}\to 0$ as $\varepsilon \to 0$, as long as $\tau >0$, since $\tilde{V}_{0,\varepsilon^{-1}}(\tau)=\frac{1}{\varepsilon^2}V_0\left(\frac{\tau}{\varepsilon (1-\tau)}\right)$ and sufficient decay is assumed on $V_0$. 

    For this problem, we consider solutions that satisfy the boundary condition at $\tau =1$. In particular, for $\mu\ne 0$ we shall work with the distinguished orbit as set by the condition \eqref{eq:out-dist-munot0}, denoted $\psi_+(t,\mu,\varepsilon)$. For $\mu = 0$, we take $\psi_+(t,0,\varepsilon)$ to be the solution satisfying \eqref{eq:out-dist-mu0} found from the blow-up procedure of Section 3.  
 
 The angular equation for the outer model problem that is analogous to \eqref{eq:Lp_inner} is. 
 \begin{equation}\label{eq:Lp_outer}
    \begin{split}
        \dot{\psi} = & (\tau -1)\sin{2\psi} +\tau \left(\left(\mu^2 + \tilde{V}_{1} (\tau) \right) \cos^2 \psi -\sin^2 \psi \right),\\\dot{\tau} = & \tau(1-\tau)^2,
     \end{split}
\end{equation}
where $\dot{}=\frac{d}{d\rho}$. Note that this is the equation for the flow on the invariant subspace $\varepsilon=0$ for \eqref{eq:Lp3_outeraug}. Here we have removed the $V_0$ term, which in the $\tau$ blown-up coordinates is still small for $\tau > \frac12$ since by (A1) we have that $\varepsilon^{-2} V_0(x/\varepsilon) \leq C_0 \varepsilon |x|^{-3-\gamma}.$

The base solution for the outer problem, which we use to estimate the distinguished orbit, will satisfy the first equation of \eqref{eq:Lp_outer}. It is constructed using the asymptotic (boundary) condition \eqref{eq:out-dist-munot0} for $\mu\ne 0$, and the blow-up procedure with \eqref{eq:out-dist-mu0}. This is analogous to the $\varepsilon >0$ construction and, in fact, we observe that $\psi_0(t,\mu)=\psi_+(t,\mu,0)$ for all $t\in(-\infty,+\infty)$.

 We need to establish the appropriate estimate at the threshold. In the current set of variables, that occurs when $t=t_\varepsilon$ where $t_\varepsilon=\varepsilon^\kappa+\kappa\ln\varepsilon$, and \correction{comment 18}{
\begin{equation}
    \label{eqn:kappadef}
    0<\kappa = \alpha + 1<1
\end{equation}
recalling the setting of $\alpha$ by $-1<\alpha <0$, see \eqref{sigma_threshold}. }

The goal of this subsection is then to prove the following lemma.

\begin{lemma}\label{lem:outerest}
If $\mu=0$, or $\mu>0$ and not in $\sigma(\Delta-V_1)$, then
there is a constant $C_1>0$ and an $\varepsilon_0 >0$, which may depend on $\mu$, so that

\begin{equation}\label{eq:outest}\left |\psi_+(t_\varepsilon,\mu,\varepsilon) -\psi_0(t_\varepsilon,\mu )\right|\le C_1 \varepsilon
\end{equation}
when $\varepsilon \le \varepsilon_0$. 

\end{lemma}
\begin{proof}
As in the proof of Lemma \ref{main_inner_estimate}, we will show that an estimate like \eqref{eq:outest} holds at successive values of $t$, culminating at $t_\varepsilon$. In this case, $t$ starts large and decreases down to $t_\varepsilon$. To that end, set
$$\nu_+(t,\mu,\varepsilon)=\psi_+(t,\mu,\varepsilon)-\psi_0(t,\mu).$$

The equation satisfied by $\nu_+$ is
\correction{comment 19}{
\begin{equation}\label{etaone_tau} 
\dot{\nu}_+ = \left[2(\tau-1)\cos 2\psi_0-\tau (\tilde{V}_1(\tau)+\mu^2+1)\sin2\psi_0\right]\nu_++\frac{\tau}{\varepsilon^2}\tilde{V}_0(\tau)+O(\nu_+^2).
\end{equation}}

For the outer problem, we denote the relevant sections by 
$$\Omega_i=\{(\psi,\tau_i,\varepsilon):0\le \varepsilon \le \varepsilon_i\},$$
where $\tau_i=\tau(t_i)$, and $0\le \tau_1 \le \tau_0<1$, with $\tau_1$ being close to $\tau=1$.

 It is again divide into three steps.

\noindent {\bf Step 1} The first step uses the relevant invariant manifolds to show that 
$$|\nu_+(t_1,\mu,\varepsilon)|\le C\varepsilon$$ 
for $\mu\ge 0$ and $\mu\notin \sigma_p(\Delta-V_1)$. The argument will be different depending on whether $\mu=0$ or $\mu>0$. 
First, consider $\mu >0$ and fixed. The (local) center manifold at $(\arctan(-\mu),1,0)$, $W^{c}_{loc}$, for \eqref{eq:Lp3_outeraug} is $2$-dimensional and can be viewed as the graph of a $C^1$ function $h^{c}$ from $$\{ (\tau,\varepsilon):\tau_0\le \tau \le 1, 0\le \varepsilon \le \varepsilon_0\}$$ to $\psi \in \mathbb{R}$, for some $\tau_1$ sufficiently close to $1$ and $\varepsilon_0$ sufficiently small. 
Since $W^{c}_{loc}$ is positively invariant and the complementary direction is (exponentially) repelling, it is unique, and hence the curves $\hat{\psi}_+(t,\mu,\varepsilon)=(\psi(t,\mu,\varepsilon),\tau,\varepsilon)$, for $0<\varepsilon<\varepsilon_0$ must lie in it for large enough positive $t$. Moreover, $h^{c}$ is $C^1$ and $\psi_-(t,\mu,\varepsilon)=h^{c}(\tau (t),\varepsilon)$. Recall that $\mu>0$ is fixed here. It then follows that 
$\left| \nu_+(t,\varepsilon)\right| \le C \varepsilon,$
for some constant $C>0$, $\varepsilon \le \varepsilon_0$ and large positive $t$. In particular it holds for $t=t_0$, where $\tau(t_0)=\tau_0$ and $t_0$ is chosen large enough negative.

For the case $\mu=0$, we have proved the necessary estimate in Lemma~\ref{lem:dist-lipschitz}. Note that the argument is similar, but uses the center-stable manifold for the $\mu=0$ case, and is more complicated as it needed the blow-up procedure. 
\newline {\bf Step 2}
The flow of \eqref{eq:Lp3_outeraug} is $C^1$ and hence, by the Implicit Function Theorem, there is a $C^1$ transition map from $\Omega_0$ to $\Omega_1$, taking $\varepsilon_0$ smaller if necessary. From applying this map, it follows that $\left|\nu_+(s_1,\varepsilon)\right|\le C\varepsilon$, for some constant $C>0$ and sufficiently small $\varepsilon$. 
 This will be true up to any $t_1$, provided that lies in $(0,1)$ and is independent of $\varepsilon$. A value of $\rho=\rho_1$ will be set in the next section and $t_1=\rho_1 +\ln \rho_1$.
 \newline \noindent {\bf Step 3} As in the previous section, we will switch to expressing $\nu_+$ in terms of the original radial variable, which, for the outer problem scaling, is $\rho$, and use the notation $\eta(\rho)=\nu_+(\rho + \ln \rho)$. Recall that $\rho=\varepsilon r$, and \correction{comment 20}{$ \kappa =\alpha + 1<1$} makes the thresholds for the inner and outer problems match. 
The equation for $\eta(\rho)$ is
\begin{equation}\label{etaone_rho}
\eta'=-\left[\frac2\rho\cos 2\psi_0+ (\tilde{V}_1(\tau)+\mu^2+1)\sin2\psi_0\right]\eta+\frac{1}{\varepsilon^2}V_0(\frac{\rho}{\varepsilon})+O(\eta^2),
\end{equation}
where $\eta'=\frac{d\eta}{d \rho}$. 

We will need information on the asymptotics of the solution of the model outer problem \eqref{eq:Lp_outer} as $\rho \to 0$, or, equivalently, as $t \to -\infty$.

If $\mu \notin \sigma_p(\Delta-V_1)$ then 
$$\psi_1( t,\mu)\to (2k+1)\frac{\pi}{2}$$ for some integer $k$ as $t \to -\infty$.

This follows from considering the dynamics of \eqref{eq:Lp_outer} near $\tau =0$ and the fact that tending to $(2k+1)\frac{\pi}{2}$ corresponds to the boundary condition at $\rho=0$ not being satisfied, as it should not be when $\mu$ is not an eigenvalue.

From Step 2 above, there are an $\varepsilon_0$ and a constant $ C_1>0$ independent of $\varepsilon$ so that
    $$\left|\eta(\rho_1)\right| \le C_1 \varepsilon,$$
    when $\varepsilon \le \varepsilon_0$, and \correction{comment 24}{$\gamma$ is as in (A1).}

With $\mu\ge0$, we assume that $\mu \notin \sigma_p \left(\Delta - V_1\right)$, and we need to show that there is a $C>0$ and $\varepsilon_0 >0$ so that
    \begin{equation}
    \left | \eta(\rho_\varepsilon)\right | \le C \varepsilon,
    \end{equation}

 We use the same methodology as for the inner system by expressing $\eta$ as a solution of an integral equation version of \eqref{etaone_rho}. The first step is to obtain estimates on the integrating factor. 

Proceeding as before, we will need estimates on $$g(\rho)=\frac2\rho\cos 2\psi_0+ (\tilde{V}_1(\tau)+\mu^2+1)\sin2\psi_0,$$ which is the parenthetical quantity in \eqref{etaone_rho}.
With $\mu \notin \sigma_p \left(\Delta - V_1\right)$, the base solution $\psi_0$ tends to $\frac\pi2+k\pi$, as $\rho\to 0$ for some integer $k$. It follows that $\cos2\psi_0 \to -1$ and $\sin 2 \psi_0 \to 0$ as $\rho \to 0$. Since this is independent of $\varepsilon$, for fixed $0< \delta< \frac14$, \correction{comment 27}{let us take   $\rho_1> 0$} small enough so that 
\begin{equation}\label{estimate_intfac_tau}
\frac{-2-\delta}{\rho} < g(\rho) < \frac{-2+\delta}{\rho},
\end{equation}
 holds for $\rho\le\rho_1$.   
\correction{comment 28}{By integrating (backwards) from $\rho_1$ to $\rho < \rho_1$}, we obtain
$$\ln\left(\frac{\rho_1}{\rho}\right)^{2+\delta}>\int_{\rho_1}^\rho g(z)dz>\ln\left(\frac{\rho_1}{\rho}\right)^{2-\delta},$$ 
and exponentiating, with $\Gamma (\rho)=\exp\{\int_{\rho_1}^\rho g(z)dz\}$ being the integrating factor we want,
$$\left(\frac{\rho_1}{\rho}\right)^{2-\delta}<\Gamma (\rho)<\left(\frac{\rho_1}{\rho}\right)^{2+\delta}.$$

\correction{comment 30}{Multiplying \eqref{etaone_rho} by $\Gamma(\rho)$ and integrating from $\rho_1$ to $\rho$}, we obtain an expression for $\eta$
\begin{equation}\label{int_eq_outer}
    \eta(\rho)=\Gamma^{-1}(\rho)\Gamma(\rho_1)\eta(\rho_1)+\Gamma^{-1}(\rho)\int_{\rho_1}^\rho \Gamma(z)h(z,\eta,\varepsilon)dz,
\end{equation}
where $h(z,\eta,\varepsilon)=\frac{1}{\varepsilon^2}V_0\left(\frac{\rho}{\varepsilon}\right)+O(\eta^2)$. From the assumption (A1) on the decay of $V_0$, and recalling that $\rho_\varepsilon =\varepsilon^\kappa$, we can estimate
\begin{equation}\label{potential_estimate}
    \left|\frac{1}{\varepsilon^2}V_0\left(\frac{\rho}{\varepsilon}\right)\right|\le C \varepsilon^{(1-\kappa)(3+\gamma)-2},
\end{equation}
provided $\rho \ge \rho_\varepsilon$.

We also have
\begin{equation}\label{intfac_estimate}
    \Gamma^{-1}(\rho)\int_\rho^{\rho_1}\Gamma(z)dz\le\frac{\rho_1}{1+\delta},
\end{equation}
for any $\rho\le\rho_1$.

\correction{comment 31}{Again, we make a bootstrap argument and assume that $\left|\eta(\rho)\right|\le 2 C_1\varepsilon$ for $0< \delta < \frac{\gamma}{4}$ with $\gamma > 0$ as in \eqref{eqn:assump1}.}  For all $\rho\in [\rho_\varepsilon,\rho_1]$. 
Plugging \eqref{potential_estimate} and \eqref{intfac_estimate} into \eqref{int_eq_outer} and estimating, we obtain
\begin{equation}\label{etaone_outer_estimate}
    \left|\eta(\rho)\right| \le \left|\eta(\rho_1)\right|+C\frac{\rho_1}{1+\delta}\left[\varepsilon^{(1-\kappa)(3+\gamma)-2}+\varepsilon^{ 2}\right] 
\end{equation}
for some (new) constant $C>0$.

Determine $C_1$ such that $\left|\eta(\rho_1)\right| \le \frac14 C_1 \varepsilon$.  For the sake of bootstrapping, let us assume that $\left|\eta(\rho)\right| \le  C_1 \varepsilon$ for $\rho_1 >\rho> \rho_\varepsilon$.
Plugging in, we will require $(1-\kappa)(3+\gamma) -2 > 1$ or $(1-\kappa) (3 + \gamma) > 3$, i.e. $\kappa < \frac{\gamma}{3+\gamma}$.  Since we may take $\kappa$ very close to $0$, the result gives that for $\varepsilon$ sufficiently small
\begin{equation}\label{etaone_outer_estimate1}
    \left|\eta(\rho)\right| \le \frac12 C_1 \varepsilon,
\end{equation}
which closes the bootstrap.
\end{proof}

\section{Proof of the eigenvalue count}\label{sec:proof}

The eigenvalue count as stated in Theorem~\ref{thm:1} will be divided into two parts. There are $m(V_1)$ of $\mathcal{O}(\varepsilon^2)$. The remaining $m(V_0)$ are all $\mathcal{O}(1)$. We first deal with the harder part which concerns the $\mathcal{O}(\varepsilon^2)$ eigenvalues. 

The strategy is topological in that we focus on values of $\mu$ at which matching does not occur, i.e., away from eigenvalues, and keep the inner and outer orbits separated, and then count the matchings in between. It is essentially an application of Sturm-Liouville Theory expressed in terms of the rotation of the angular variable.

\subsection{Separation of distinguished orbits at threshold}

A key point is to show that the distinguished orbits do not match with each other at the end-points of the $\mu$-interval. There is, however, a preliminary issue to resolve as we want to match the inner orbits with outer ones at the threshold, but the inner and outer orbits are expressed in different variables.
The transformation from $\theta$ to $\psi$ is given by 
\begin{equation}\label{eq:theta-psi}
\psi=\Psi(\theta,\varepsilon)=\arctan\left(\frac1\varepsilon \tan\theta\right)
\end{equation}
The issue is clear from \eqref{eq:theta-psi} as if $\tan \theta$ is $\mathcal{O}(1)$, then there is not much control on $\psi$. In the following lemma, we work with the wrapped version of $\arctan$ which takes values in $S^1=[-\frac\pi2,\frac\pi2]/(-\frac\pi2 \sim \frac\pi2)$.

We will need a further restriction on $\alpha$, i.e., the threshold value, recalling that $s_\varepsilon = \varepsilon^\alpha +\alpha \ln \varepsilon$.

\begin{lemma}\label{lem:theta-psi-control}
Take $\alpha <-\frac12$, then there is a $K>0$, independent of $\varepsilon$, and an $\varepsilon_0>0$ so that 
\begin{equation}\label{eq:sep-model-inner}
|\Psi(\theta_0(s_\varepsilon), \varepsilon)-\frac\pi2| >K, \end{equation} and 
\begin{equation}\label{eq:sep-full-inner}
    |\Psi(\theta_-(s_\varepsilon,\mu, \varepsilon), \varepsilon)-\frac\pi2| >K
    \end{equation}
in $S^1$ for any $\varepsilon \in (0,\varepsilon_0)$, and any $\mu$ bounded.
    
\end{lemma}
\begin{proof} The inequality~\eqref{eq:sep-model-inner} follows easily from Lemma~\ref{Lemma:theta00} provided $\alpha<-\frac12$. 
The second inequality~\eqref{eq:sep-full-inner} will follow if we can show there is a constant $C>0$ so that \begin{equation} \label{eq:sep-proof-ineq}   
|\tan\theta_-(s_\varepsilon,\varepsilon)|<C\varepsilon,\end{equation} for sufficiently small $\varepsilon$. Adding and subtracting $\tan \theta_0$. Applying the triangle inequality. $$|\tan \theta_-|\le |\tan \theta_--\tan \theta _0|+|\tan \theta_0|,$$ from which we can conclude that there is a constant $C>0$, and
$$|\tan\theta_-|\le C \left(|\theta_--\theta _0| + |\theta_0|\right),$$ where we have used the Mean-Value Theorem and the fact that both $\theta_-$ and $\theta_0$ are close to $0$. Inequality~\eqref{eq:sep-proof-ineq} then follows from Lemmas \ref{Lemma:theta00} and \ref{main_inner_estimate}, using that $\alpha < -\frac12$.

\end{proof}

\begin{lemma}\label{lem:outer-base-control}
If $\mu=0$ or $\mu=\tilde{\mu}$ and $\tilde{\mu}\notin \sigma(\Delta-V_1)$ then there is a constant $C>0$, and $\varepsilon_0>0$ so that
\begin{equation}\label{sep-mod-outer}|\psi_0(t_\varepsilon,\mu)-\frac\pi2| < C\varepsilon^\kappa,\end{equation} and
\begin{equation}
    \label{sep-full-outer}|\psi_+(t_\varepsilon,\mu,\varepsilon)-\frac\pi2| < C\varepsilon^\kappa,
\end{equation}
for $\varepsilon<\varepsilon_0$, and recalling $0<\kappa=\alpha + 1<\frac12$.
\end{lemma} 

\begin{proof} We know the asymptotics as $t\to-\infty$ of the outer base solution $\psi_0(t,\mu)$ provided $\mu$ is not an eigenvalue or $\mu=0$, since $0$ is taken to not be an end-point resonance. Indeed, in either of these cases
\begin{equation}\label{}\psi_0(t,\mu)\to \frac\pi2,\end{equation}
as $t\to-\infty$. The estimate \eqref{sep-mod-outer} in the  lemma then follows from looking at the linearization of \eqref{eq:Lp_outer} at $\psi=\frac\pi2, \tau=0$.Note that $C$ and $\varepsilon_0$ may depend on $\tilde{\mu}$, and that $\kappa \in (0,\frac12)$ from our new requirement on $\alpha$. 

The second inequality \eqref{sep-full-outer} is derived as in the proof of Lemma~\ref{lem:theta-psi-control} by adding and subtracting the base solution, applying the triangle inequality and invoking Lemma~\ref{lem:outerest}.
\end{proof}

From Lemmas~\ref{lem:theta-psi-control} and \ref{lem:outer-base-control}, we can see that the outer and inner distinguished orbits, in both the model and full problems, stay separated at threshold provided $\mu$ is not at $0$ or an eigenvalue of $\Delta - V_1$. This can be summed up in the following corollary to Lemmas~\ref{lem:theta-psi-control} and \ref{lem:outer-base-control}.

\begin{corollary}\label{cor:sep}
If $\alpha<-\frac12$ and $\mu=0$ or $\mu >0$ and $\notin \sigma(\Delta-V_1)$, then, there is an $\varepsilon_0$ so that
\begin{equation}
   \Psi(\theta_0(s_\varepsilon), \varepsilon) \ne \psi_0(t_\varepsilon,\mu), 
\end{equation}
and
\begin{equation}\label{eq:no-epres}
   \Psi(\theta_-(s_\varepsilon,\mu,\varepsilon), \varepsilon) \ne \psi_+(t_\varepsilon,\mu,\varepsilon), 
\end{equation}
for $0<\varepsilon<\varepsilon_0$, and the distance between them, in each case, is bounded below independently of $\varepsilon \in(0,\varepsilon_0)$.
\end{corollary}
In particular, \eqref{eq:no-epres}, in the $\mu=0$ case, says that $0$ is not an end-point resonance of $\Delta-W_\varepsilon$ if $\varepsilon$ is sufficiently small. This verifies (1) of Theorem 1.

\subsection{Counting the small eigenvalues}\label{sec:small}

Throughout this section, we will work with angles $\theta$ (and $\psi$) in all of $\mathbb{R}$, rather than $S^1$, and use the notation $\tilde{\theta}$ to distinguish it. If the covering map is denoted by $\Pi:\mathbb{R}\to S^1$ then, $\Pi (\tilde{\theta}) =\theta$.

 We introduce new notation for the trajectories of interest in the covering space $\mathbb{R}$.   The lift of $\theta_0(s)$ will be denoted $\tilde{\theta}_0(s)$ if it satisfies $\lim_{s\to -\infty} \tilde{\theta}_0(s)=0$, called the primary lift. Other lifts will be denoted $\tilde{\theta}^k_0(s)=\tilde{\theta}_0(s)+k\pi$. The (primary) lift of $\theta _-(\lambda, s)$, given in \eqref{eq:thetahatminus}, will be denoted $\tilde{\theta}_-(\mu,s)$, where $\lambda =\varepsilon^2\mu$. The other lifts will be denoted $\tilde{\theta}^k_-(\mu,s)=\tilde{\theta}_-(\mu,s)+k\pi$. Analogously, we use the notation $\tilde{\psi}_0(\mu,t)$ and $\tilde{\psi}_+(\mu,t)$ for the primary lifts of the model and full outer problems respectively, which are required to have the limit $\arctan(-\mu)$ as $t\to +\infty$ where $\arctan$ is taken to have values in $(-\frac\pi2,\frac\pi2)$. Note we are suppressing explicit mention of the dependence of $\psi_-$ on $\varepsilon$.

 We set $\mu_0$ as above, i.e., if $\mu \in \sigma(\Delta -V_1)$ and $\mu \ge 0$ then $\mu \in (0,\mu_0)$. \correction{comment 32}{In particular, 
$\mu_0\notin\sigma_p(\Delta-V_1)$, which is used in what follows.}  
It then follows that 
\begin{equation}\label{eq:asym_model_noeig}
     \begin{split}
         \tilde{\psi}_0 &(t,\mu_0) \to\frac{\pi}{2},\hspace{.7mm} {\rm and} \\
         \tilde{\psi}_0 &(t,0) \to (2M+1)\frac{\pi}{2},
     \end{split}
 \end{equation}
 for some $M\in \mathbb{N}$, as $t \to -\infty$.
The next lemma follows from an application of Sturm-Liouville Theory. It could also be proved directly using the geometric techniques of this section. Recall that $m(V_1)$ is the number of positive eigenvalues of $\Delta-V_1$ on $L_2$ radial functions in $\mathbb{R}^3$.
 \begin{lemma}\label{lem:spectrum_outer}
 $m(V_1) = M$.
 \end{lemma}

The picture to have in mind is that as $\mu$ increases from $0$ to $\mu_0$, $\lim_{t\to-\infty}\psi_0(t,\mu)$ will decrease in increments of $\pi$, i.e., from $\frac\pi2$ to $\frac{3\pi}{2}$ and so on. At each $\mu$ value at which it jumps, there will be an eigenvalue as then $\lim_{t\to-\infty}\psi_0(t,\mu)$ will be a multiple of $\pi$. This means it is an eigenvalue as then the regularity condition at $r=0$ ($\rho=0$) is satisfied as well as the decay condition as $r\to +\infty$ ($\rho\to+\infty$).

We can now complete the proof of the main result for eigenvalue parameter values of size $O(\varepsilon^2)$. 
Instead of letting $t$ tend to $-\infty$ as in Lemma~\ref{lem:spectrum_outer}, we evaluate the outer distinguishing orbit at the time it crosses the threshold, namely $t=t_\varepsilon$, which is close to $-\infty$, in this scaling. The condition for $\mu$ to be an eigenvalue is that there is a match with the distinguished orbit of the inner system:
\begin{equation}\label{eq:eig-condition-wrapped}
     \Psi(\theta_-(s_\varepsilon,\mu,\varepsilon)) = \psi_+(t_\varepsilon,\mu),
\end{equation}
in $S^1$.
We need to lift \eqref{eq:eig-condition-wrapped} to the covering space $\mathbb{R}$. Set $$\tilde{\Psi}(\theta_-^k)=\Psi(\theta_-^k)+k\pi.$$ Then $\mu$ is an eigenvalue if there is a $k\in \mathbb{N}$ so that
\begin{equation}\label{eq:eig-condition-unwrapped}
     \tilde{\Psi}(\theta_-^k(s_\varepsilon,\mu,\varepsilon) )= \tilde{\psi}_+(t_\varepsilon,\mu,\varepsilon),
\end{equation}
Now set
\begin{equation}\label{Sigma_k}
    \Sigma^k(\mu,\varepsilon)=\tilde{\Psi}(\theta_-^k(s_\varepsilon,\mu,\varepsilon)) - \tilde{\psi}_+(t_\varepsilon,\mu,\varepsilon)
\end{equation}
If $\Sigma^k(\mu,\varepsilon)=0$ for any integer $k$, then $\lambda =\varepsilon^2\mu^2$ is an eigenvalue of $\Delta-(V_0+V_{1,\varepsilon})$ for that value of $\varepsilon$. 

We claim there are, at least, $M$, as in \eqref{eq:asym_model_noeig}, values of $k$ for which $\Sigma^k(\mu,\varepsilon)=0$ at some value of $\mu\in(0,\mu_0)$ if $\varepsilon$ is sufficiently small. This is an application of the Intermediate Value Theorem, on the threshold line as $\mu$ varies from $0$ to $\mu_0$, given the following four facts. Their proofs are articulated in the sentences after each one.
\begin{enumerate}
    \item $\Sigma^k(\mu,\varepsilon)$ is continuous in $\mu$ on $[0,\mu_0]$ if $\varepsilon$ is sufficiently small. This follows straightforwardly from continuity in paramters for an ODE, except at $\mu=0$. The continuity there comes from Lemma~\ref{lem:contorbit}. 
    \item There is a constant $0\le C_-<\frac\pi2$ so that $$\theta_-|(s_\varepsilon,\mu,\varepsilon)|\le C_-,$$ for all $\mu\in(0,\mu_0)$ and $\varepsilon$ sufficiently small.  This follows from Lemma~\ref{lem:theta-psi-control}.
    \item There is a constant $C_1>0$ so that $|\tilde{\psi}_+(t_\varepsilon,0,\varepsilon)-(2M+1)\frac\pi2|<C_1 \varepsilon^\kappa$ if $\varepsilon$ is sufficiently small. This comes from \eqref{eq:asym_model_noeig} and Lemma~\ref{lem:outerest} using the triangle inequality. 
    \item There is a constant $C_2>0$ so that $|\tilde{\psi}_+(t_\varepsilon,\mu_0,\varepsilon)-\frac\pi2|<C_1 \varepsilon^\kappa$ if $\varepsilon$ is sufficiently small. This also comes from \eqref{eq:asym_model_noeig} and \ref{lem:outerest} using the triangle inequality. 
\end{enumerate}

It remains to prove that there are at most $m(V_1)$ such small eigenvalues, i.e., that the eigenvalue count is exact.

\begin{lemma}\label{lem:exact_count}
If $\varepsilon$ is sufficiently small and $\mu=\hat{\mu}$ is a zero of $\Sigma^k$, then 
$$\frac{\partial}{\partial\mu}
\biggr{\rvert}_{\mu=\hat{\mu}} \Sigma^k(\mu,\varepsilon) >0.$$
\end{lemma}
\begin{proof}
    We first show that $\frac{\partial}{\partial\mu}
\tilde{\theta}_+(\mu, t)<0$ for any $\mu$ and $t$. It is equivalent to show that $\frac{\partial}{\partial\mu}
\tilde{\psi}_+(\mu, t)<0$ for any $\mu$ and $t$. As $t\to + \infty$, $\tilde{\psi}_+(\mu, t) \to \arctan (-\sqrt{\mu})$. Since $\frac{d}{d\mu}\arctan (- \sqrt{\mu})<0$, $\frac{\partial}{\partial\mu}\tilde{\psi}_+(\mu, t)<0$ for $t$ large. To show that this quantity stays negative, we calculate from \eqref{eq:Lp6}:
\begin{equation}\label{eq:variation_mu}
  \frac{d}{dt}\left({\frac{\partial\psi}{\partial\mu}} \right)=R(\tau,\psi,\mu)\frac{\partial \psi }{\partial\mu}+2 \mu \tau \cos^2\psi,  
\end{equation}
where $R=2(\tau-1)\cos 2 \psi -\tau \left((\mu^2+\tilde{V}_{0,\varepsilon}(\tau)+\tilde{V}_1(\tau)+1)\sin 2 \psi\right)$. Since $\tau\cos^2\psi \ge 0$, solving \eqref{eq:variation_mu} using an integrating factor, we can conclude that $\frac{\partial}{\partial\mu}
\tilde{\psi}_+(\mu, t)<0$ for any $\mu$ and $t$.

Similarly, we show that \correction{comment 39}{$\frac{\partial}{\partial\mu}
\tilde{\theta}_-(\mu, t)>0$} for any $\mu$ and $t$ by first showing that $\frac{\partial}{\partial\mu}
\tilde{\theta}_-(\mu, t)>0$ for large negative $t$. Using that the unstable manifold of $(0,0)$ in \eqref{eq:Lp3} is tangent to the unstable eigenvector, we can expand $$\theta =\left( \mu^2+V_0(0)+\varepsilon^2V_1(0)-1\right)\sigma +O(\sigma^2),$$ from which we can conclude that $\frac{\partial}{\partial\mu}
\tilde{\theta}_-(\mu, t)>0$ if $t\ll 0$. A similar argument using a version of \eqref{eq:variation_mu} in terms of $\theta$ then shows that $\frac{\partial}{\partial\mu}
\tilde{\theta}_-(\mu, t)>0$ for any $\mu$ and $t$.
\end{proof}

This ostensibly still allows for these eigenvalues to have multiplicity greater than $1$. However, from  classical Sturm-Liouville theory, see for instance \cite{reed1978iii}, Theorem XI.53, these eigenvalues must be simple in the class of radial functions. 

The results above complete the proof of Theorem 2. To prove (2) of Theorem 1, we need to find the remaining eigenvalues.

 \subsection{The case of $O(1)$ eigenvalues} The eigenvalues that do not tend to $0$ with $\varepsilon$ come from the inner potential $V_0$ and it is considerably easier to capture them as the outer potential manifests as just a small perturbation. \correction{comment 41}{As we need uniform bounds in $\varepsilon$, we treat this problem perturbatively and work with the solution $\theta_-(\lambda,s,\varepsilon)$ of \eqref{eq:Lp3} on the unstable manifold of $(\theta,\sigma) = (0,0)$.  We then compare it with the analogous solution of with a model problem based on just $V_0$.} Whereas the model problem in the case of $O(\varepsilon^2)$ had the eigenvalue parameter evaluated at $0$, we need to include it here for all positive values:
 \begin{equation}\label{eq:inner_model_evalue}
    \begin{split}
        \dot{\theta} = & (\sigma -1)\sin{2\theta} +\sigma \left(\lambda + \tilde{V}_{0} (\sigma)\cos^2 \theta -\sin^2 \theta \right),\\\dot{\sigma} = & \sigma(1-\sigma)^2.
    \end{split} 
 \end{equation}
 Let $\theta_0(s,\lambda)$ denote the solution of \eqref{eq:inner_model_evalue} on $W^u(0,0)$, suitably parameterized. Its asymptotics as $s\to +\infty$ are given by the following lemma.

 \begin{lemma}\label{lem:inner_orderone}
 If \correction{comment 33}{$\lambda \notin \sigma_p (\Delta-V_0)$} and $\lambda >0$, then
 \begin{equation}
     \theta_0(s,\lambda)\to \theta_{1,-}=\arctan\left({\sqrt{\lambda}}\right).
 \end{equation}
 \end{lemma}
 The purely perturbative nature of $V_{1,\varepsilon}$ then allows us to prove the lemma comparing the model and full problems for $\lambda=O(1)$.
 \begin{lemma}\label{lem:full_model_orderone}
 If $\lambda \notin \sigma_p (\Delta-V_0)$ and $\lambda >0$, then there are $\varepsilon _0$ and $C>0$ so that
 \begin{equation}
     \left | \theta_-(s,\lambda,\varepsilon)-\theta_0(s,\lambda)\right|\le C\varepsilon
 \end{equation}
 for all $s\in \mathbb{R}$.
 \end{lemma}
The new piece here is that the estimate holds as $s\to+\infty$. This can be seen by checking the bootstrapping estimates or looking at the dynamics of \eqref{eq:Lp3} in a neighborhood of $\sigma =1$ for small $\varepsilon$.

It can now be seen that the number of positive eigenvalues of $\Le$ on $L_2$ radial functions that are bounded away from zero as $\varepsilon \to 0$ is exactly $m(V_0)$ by repeating a similar argument to that used in the previous section for Lemma \ref{lem:spectrum_outer}, \correction{comment R2-7}{meaning $\theta_0(\lambda,s)$ crosses integer multiples of $\pi$ exactly $m(V_0)$ times as $\lambda$ varies over $(0,\infty)$, each crossing corresponding to an eigenvalue of $\Delta-V_0$. Since $\theta_-$ is $C\varepsilon$-close to $\theta_0$ uniformly in $s$ 
by Lemma~\ref{lem:full_model_orderone}, the same $m(V_0)$ crossings of integer multiples of $\pi$ occur for $\theta_-$ for $\varepsilon$ sufficiently small, giving exactly $m(V_0)$ eigenvalues of $\Delta-W_\varepsilon$ bounded away from zero. Simplicity and exactness of the count follow from a similar argument 
using $\frac{\partial \theta}{\partial \mu}$, as in Lemma~\ref{lem:exact_count}.}

Combining these results, we complete the proof of Theorem \ref{thm:1}, demonstrating that the total number of eigenvalues in $\sigma_p(\Delta - W_{\varepsilon})$ contained in $(0, \infty)$ is precisely $m(V_0) + m(V_1)$.

\section{Numerical Examples }\label{sec:appl}

We present three numerical scenarios that illustrate Theorem~\ref{thm:1}, demonstrating how different potential combinations affect the behavior of systems \eqref{eq:Lp3} and \eqref{eq:Lp6}. In all cases, we fix $\mu \in (0,1)$, $\varepsilon=0.1$, and use thresholds \eqref{sigma_threshold} and \eqref{tau_threshold} for $\sigma$ and $\tau$ respectively. All numerical simulations were executed on a MacBook Pro equipped with an Apple M1 chip (8-core CPU at 3.2 GHz) and 8GB of unified memory. The computations were completed with minimal processing time.\footnote{The complete MATLAB source code for all three cases is available at \url{https://github.com/efleurantin2103/AlgoSpace/tree/master/Scalar_Paper_3cases}.  The computational codes are subject to usage restrictions described in the repository documentation.}

Each scenario satisfies assumptions (A1). In particular:
\begin{itemize}
    \item All potentials are $C^\infty(\mathbb{R}^+)$ and radially symmetric;
    \item Each $V_{1,\varepsilon}$ satisfies the scaling structure described in assumption (A1);
    \item Each $V_0$ satisfies the integrability and decay assumptions required to ensure that $\Delta - V_0$ has essential spectrum $(-\infty,0]$ and a finite number of eigenvalues in $(0,\infty)$.
\end{itemize}

In all Figures in scenarios 1,2, and 3, left plots are solutions of \eqref{eq:Lp3}, while right plots are solutions of \eqref{eq:Lp6}. We are plotting the unstable manifolds of 3 different equilibrium solutions for the $\sigma$ system for scenarios 1 and 3, in particular $(0,0)$, $(0,\pi)$ and $(0, 2\pi)$ while varying $\mu \in (0,1)$. For scenario 2, in addition to the first 3 mentioned, we are also plotting the unstable manifold of $(0,3\pi)$.  For the $\tau$ system, we are plotting the center manifold of $(1, \arctan{\sqrt{\gamma}}-2\pi)$ while also varying $\mu$ from $(0,1)$.  To compute the eigenvalue counts above a given $\mu$ value, we consider $\lfloor \frac{\theta (0)-\theta(1)}{\pi} \rfloor$ in the $(\sigma, \theta)$ system and $\lfloor \frac{\psi (0)-\psi(1)}{\pi} \rfloor$ in the $(\tau, \psi)$ system for the given $\mu$.  

\subsection{Scenario 1: Exponential and Gaussian potentials}
Consider the potentials
\[
V_0(x) = -2.8e^{-x^2}, \quad V_{1,\varepsilon}(x) = -30\varepsilon^2 e^{-(\varepsilon x)^2}.
\]
The Gaussian profile ensures that $V_{1,\varepsilon}$ is rapidly decaying and scales appropriately under $\varepsilon \to 0$. Meanwhile, $V_0(x)$ is smooth, integrable on $[0,\infty)$, and satisfies $|V_0(x)| \leq C (1+x)^{-a}$ for any $a > 0$, thus fulfilling the assumptions on the background potential $V_0$. Figure~\ref{fig:Sigma-Tau1} shows the corresponding phase portraits.

\begin{figure}[tbp]
\begin{center}
\subfigure[\label{fig:Sigma-Tau1a}]{
    \includegraphics[height=8cm]{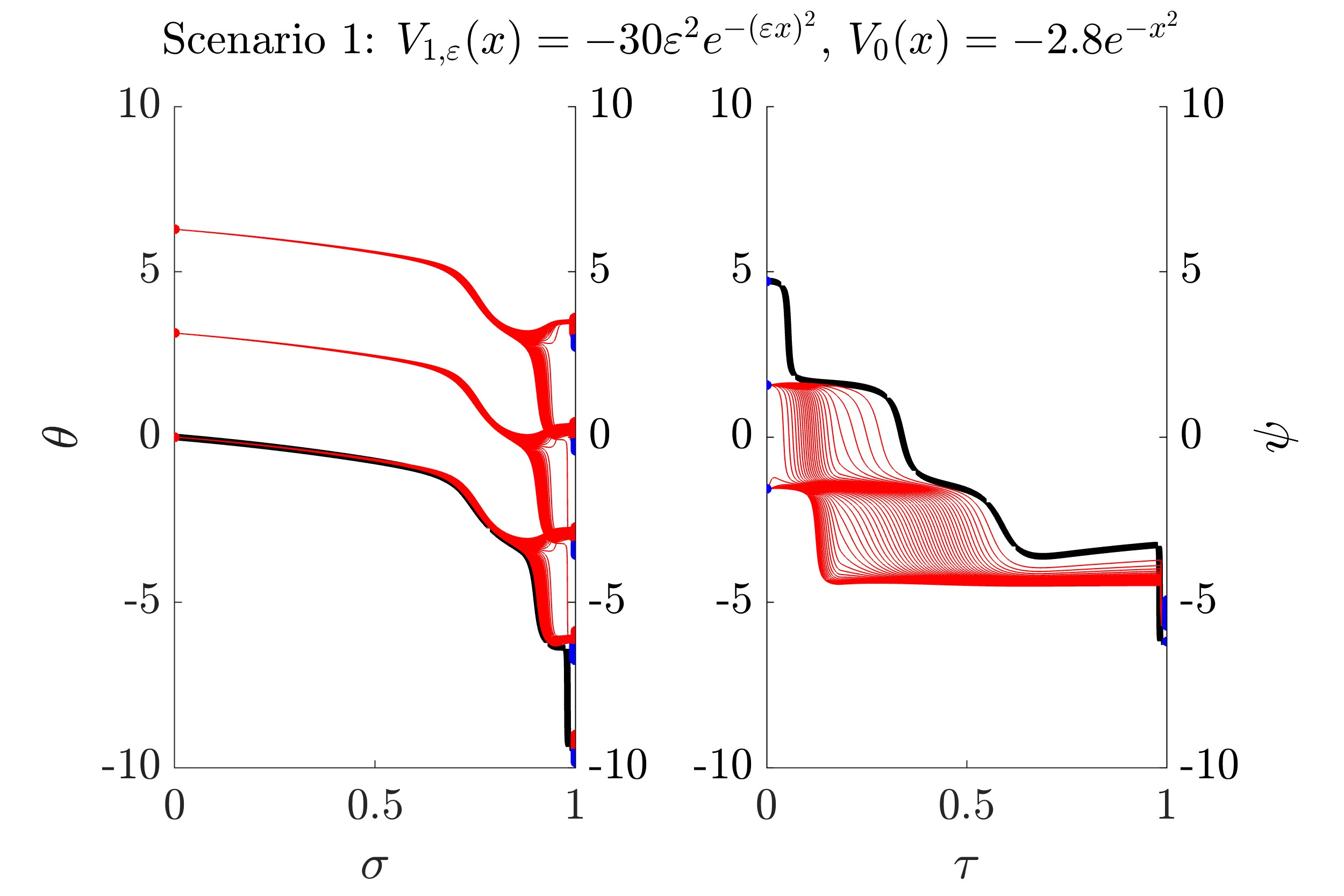}
}
\subfigure[\label{fig:Sigma-Tau1b}]{
    \includegraphics[height=8cm]{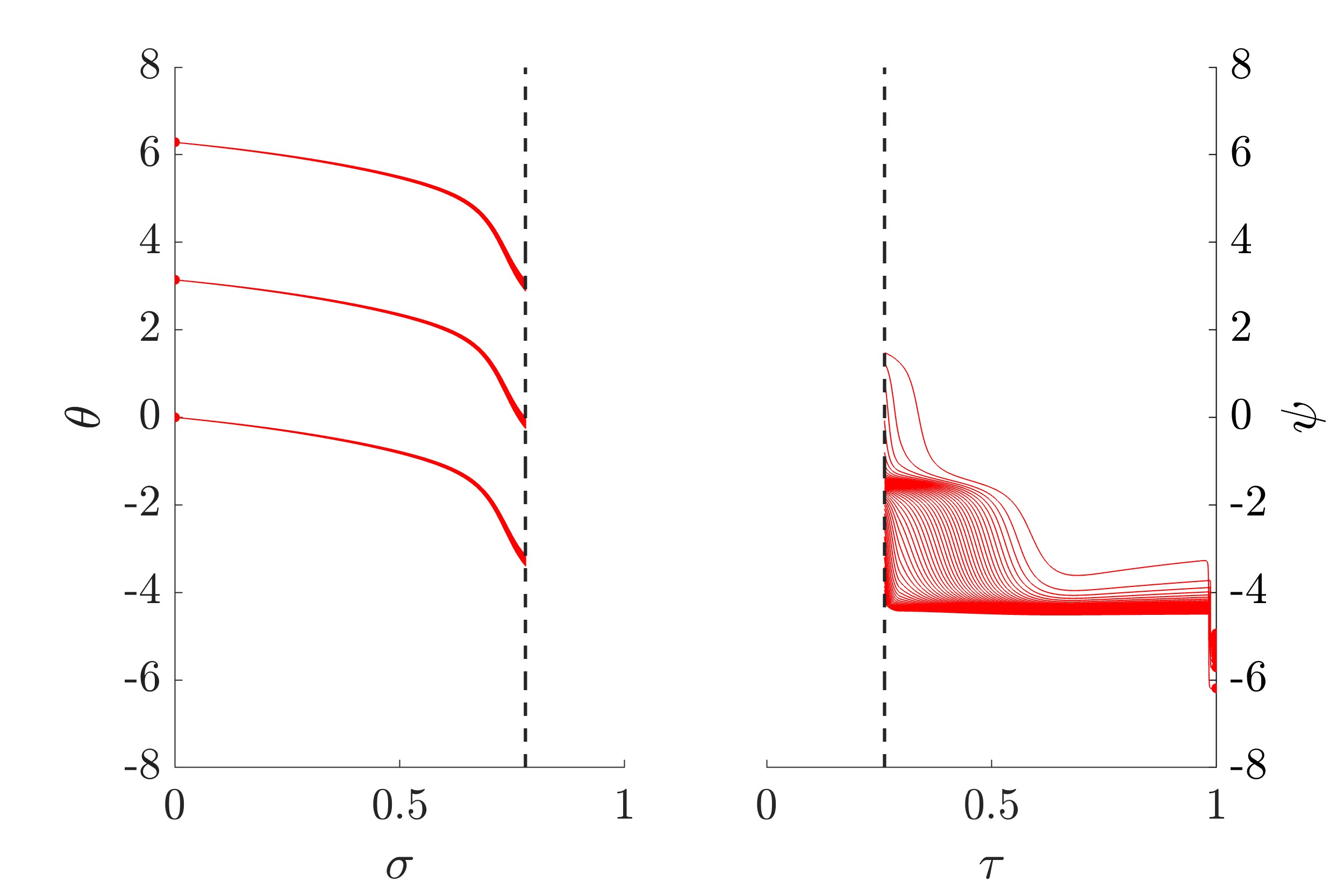}
}
\caption{Scenario 1, 
panel (a)
: Unstable manifold $W^u(\hat{\theta}_{0,-})$ (left) and center manifold $W^{c}(\hat\psi_{1,-})$ (right) from \eqref{eq:Lp3} and \eqref{eq:Lp6}. Red/blue points indicate saddle/unstable equilibria at $\sigma=0$, while blue points at $\sigma=1$ indicate center-stable equilibria. Black curves show maximum $\pi$-jumps as $\varepsilon \to 0$ for their respective systems. 
Panel (b) shows the eigenvalue counts at the 
threshold $\sigma_\varepsilon$ and $\tau_\varepsilon$: the $(\sigma,\theta)$ 
system contributes $m(V_0)=1$ and the $(\tau,\psi)$ system contributes 
$m(V_1)=2$, summing to $m(W_{\varepsilon})=3$ consistent with panel (a) and Theorem 1.}
\label{fig:Sigma-Tau1}
\end{center}
\end{figure}

\subsection{Scenario 2: Rational and hyperbolic secant potentials}

Consider the potentials
\[
V_0(x) = -\frac{2.6}{1+2x^4}, \quad V_{1,\varepsilon}(x) = -20\varepsilon^2\text{sech}(\varepsilon x).
\]
Again, $V_{1,\varepsilon}$ has the correct scaling and decay properties. The rational function $V_0(x)$ is $C^\infty$, and satisfies $|V_0(x)| \leq C (1+x)^{-4}$, which ensures square-integrability and decay sufficient to define a compact perturbation of the Laplacian. Figure~\ref{fig:Sigma-Tau2} shows the results.
\begin{figure}[tbp]
\begin{center}
\subfigure[\label{fig:Sigma-Tau2a}]{
    \includegraphics[height=8cm]{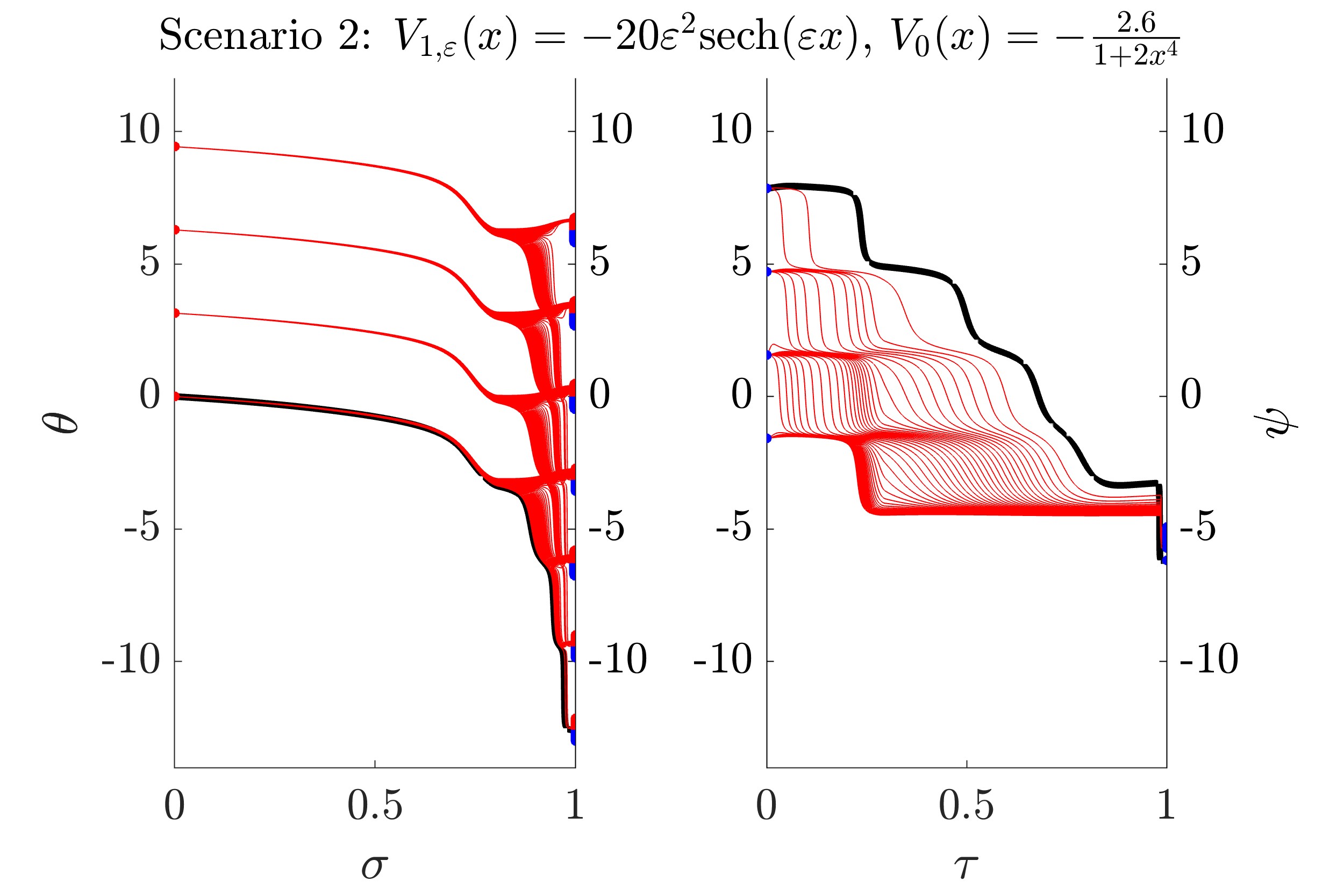}
}
\subfigure[\label{fig:Sigma-Tau2b}]{
    \includegraphics[height=8cm]{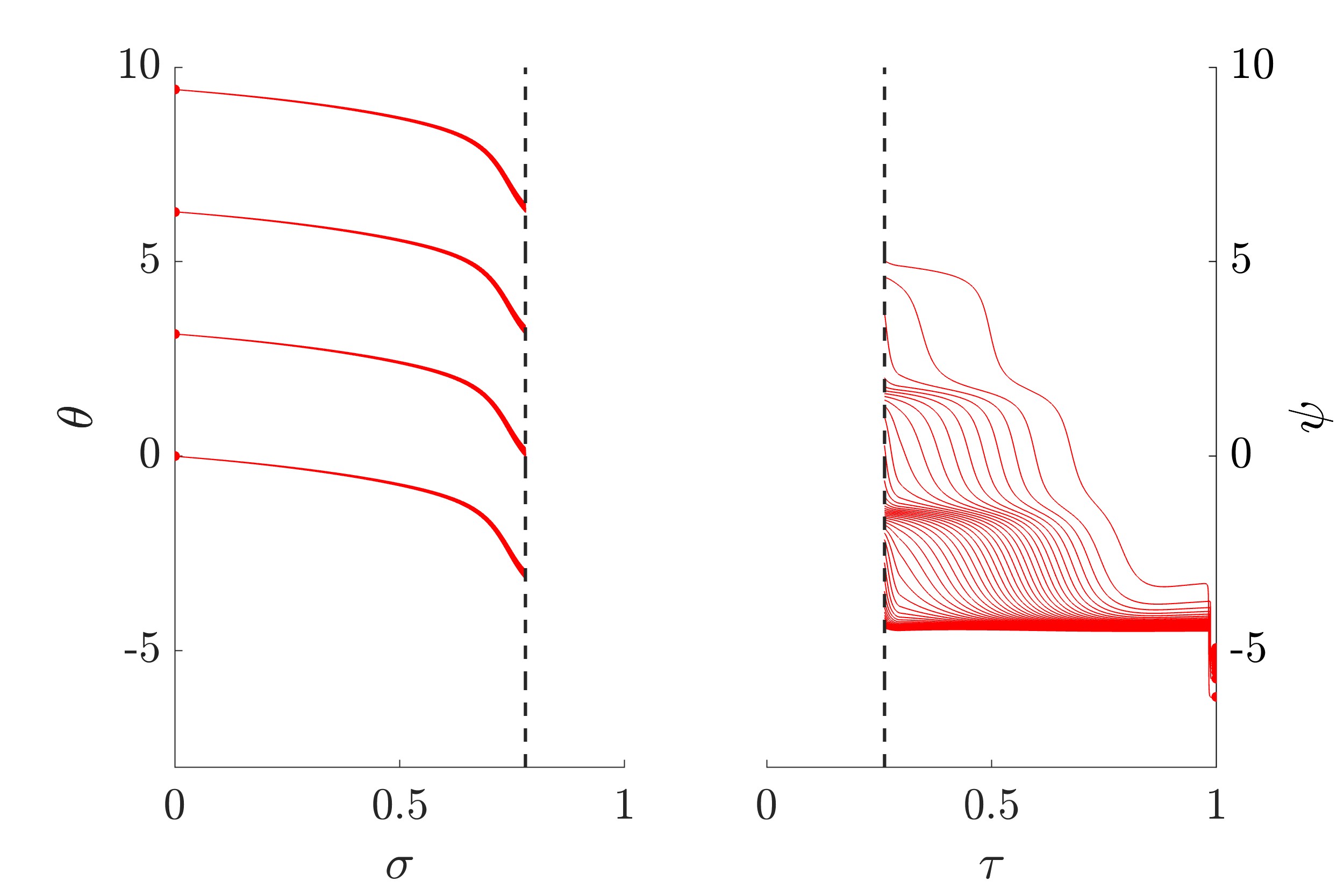}
}
\caption{Scenario 2: System behavior under hyperbolic secant and rational potentials. Panel arrangement, curve colors, and equilibria properties follow the same convention as Figure~\ref{fig:Sigma-Tau1}.}
\label{fig:Sigma-Tau2}
\end{center}
\end{figure}

\subsection{Scenario 3: Hyperbolic cosine  and rational potentials}
Consider the potentials
\[
V_0(x) = -\frac{3}{\cosh^2(1.2x)}, \quad V_{1,\varepsilon}(x) = -\frac{30\varepsilon^2}{(1+(\varepsilon x)^2)^2}.
\]

The potential $V_{1,\varepsilon}$ is smooth and scales as required, with algebraic decay. The background potential $V_0(x)$ is bounded, smooth, integrable, and satisfies $|V_0(x)| \leq 12e^{-2.4x}$, which again meets the necessary integrability and decay assumptions. Figure~\ref{fig:Sigma-Tau3} displays the associated system behavior.

\begin{figure}[tbp]
\begin{center}
\subfigure[\label{fig:Sigma-Tau3a}]{
    \includegraphics[height=8cm]{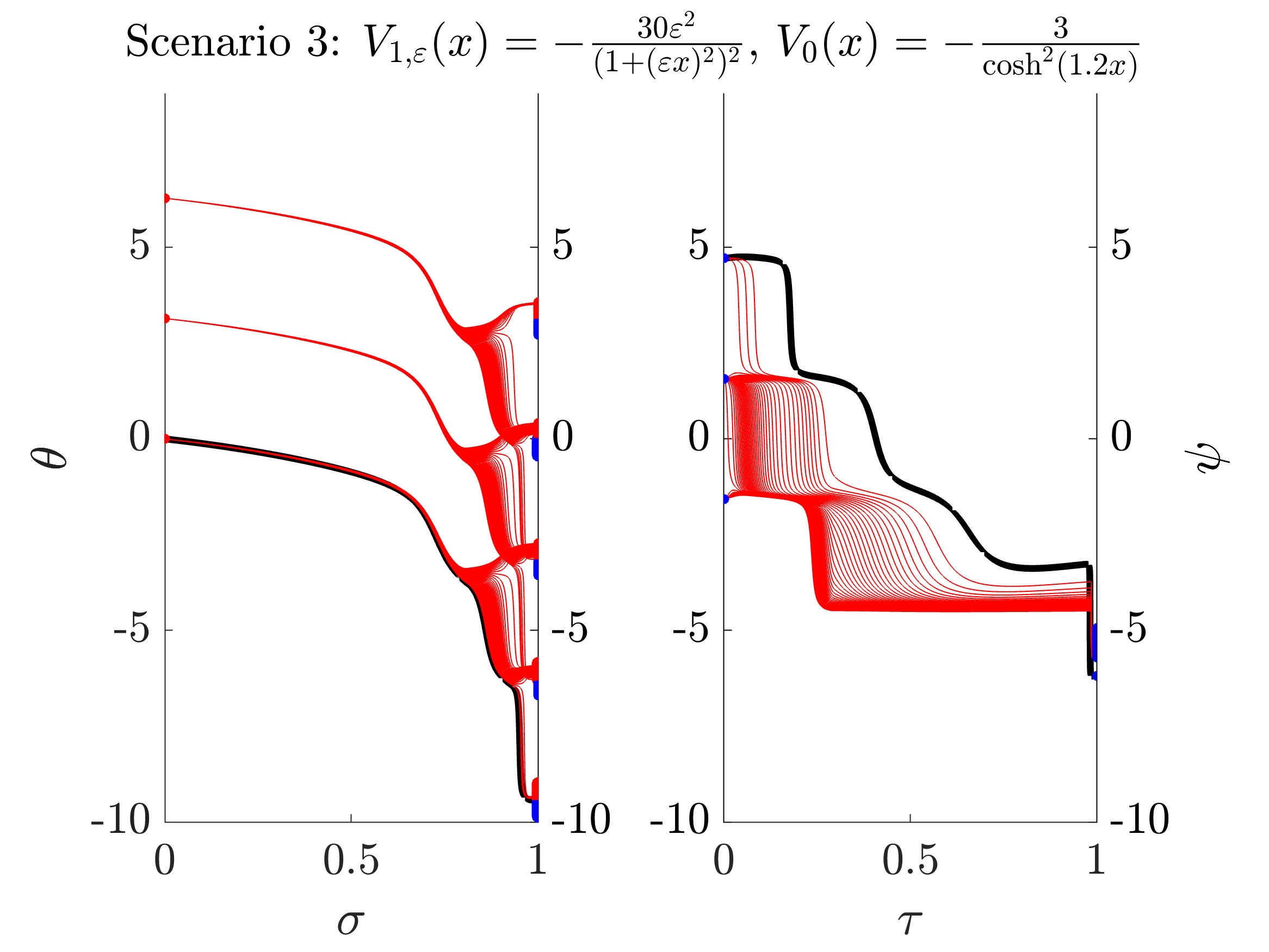}
}
\subfigure[\label{fig:Sigma-Tau3b}]{
    \includegraphics[height=8cm]{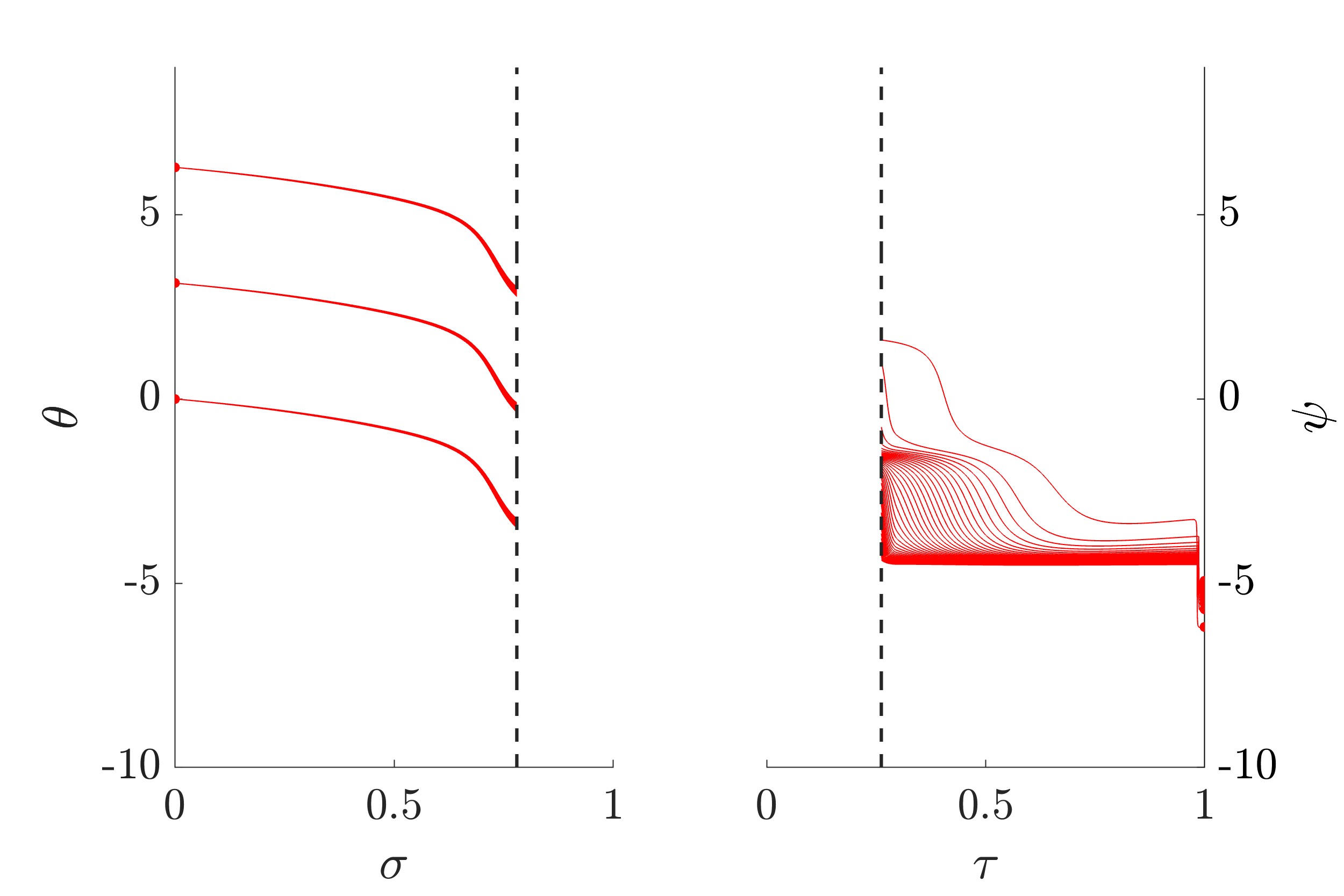}
}
\caption{Scenario 3: System behavior under rational and hyperbolic potentials. Panel arrangement and equilibria properties follow the same convention as Figure~\ref{fig:Sigma-Tau1}.}
\label{fig:Sigma-Tau3}
\end{center}
\end{figure}


\begin{figure}[tbp]
\begin{center}
\subfigure[\label{fig:Sigma-TauEa}]{
    \includegraphics[height=2.9cm]{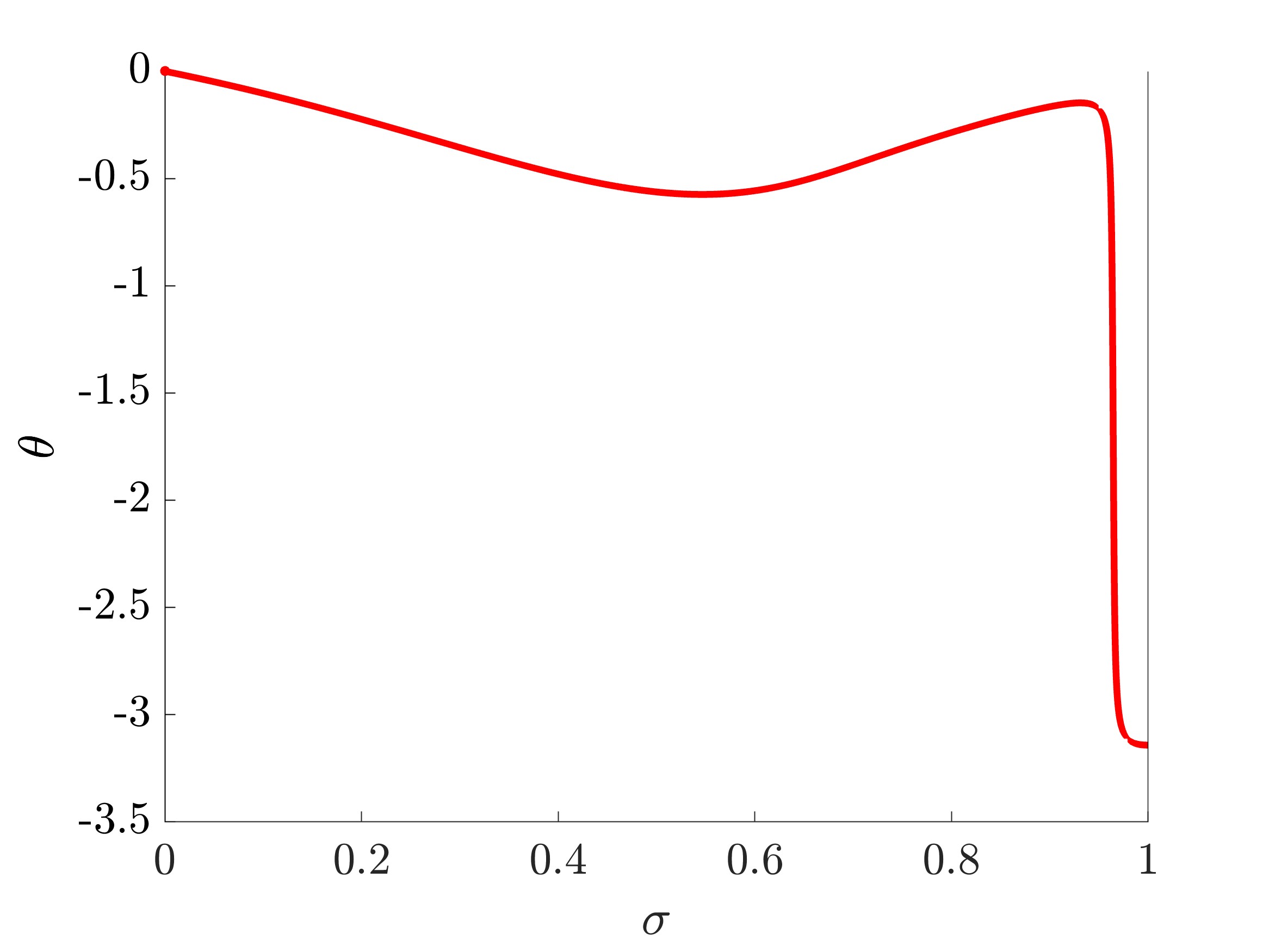}
}
\subfigure[\label{fig:Sigma-TauEb}]{
    \includegraphics[height=2.9cm]{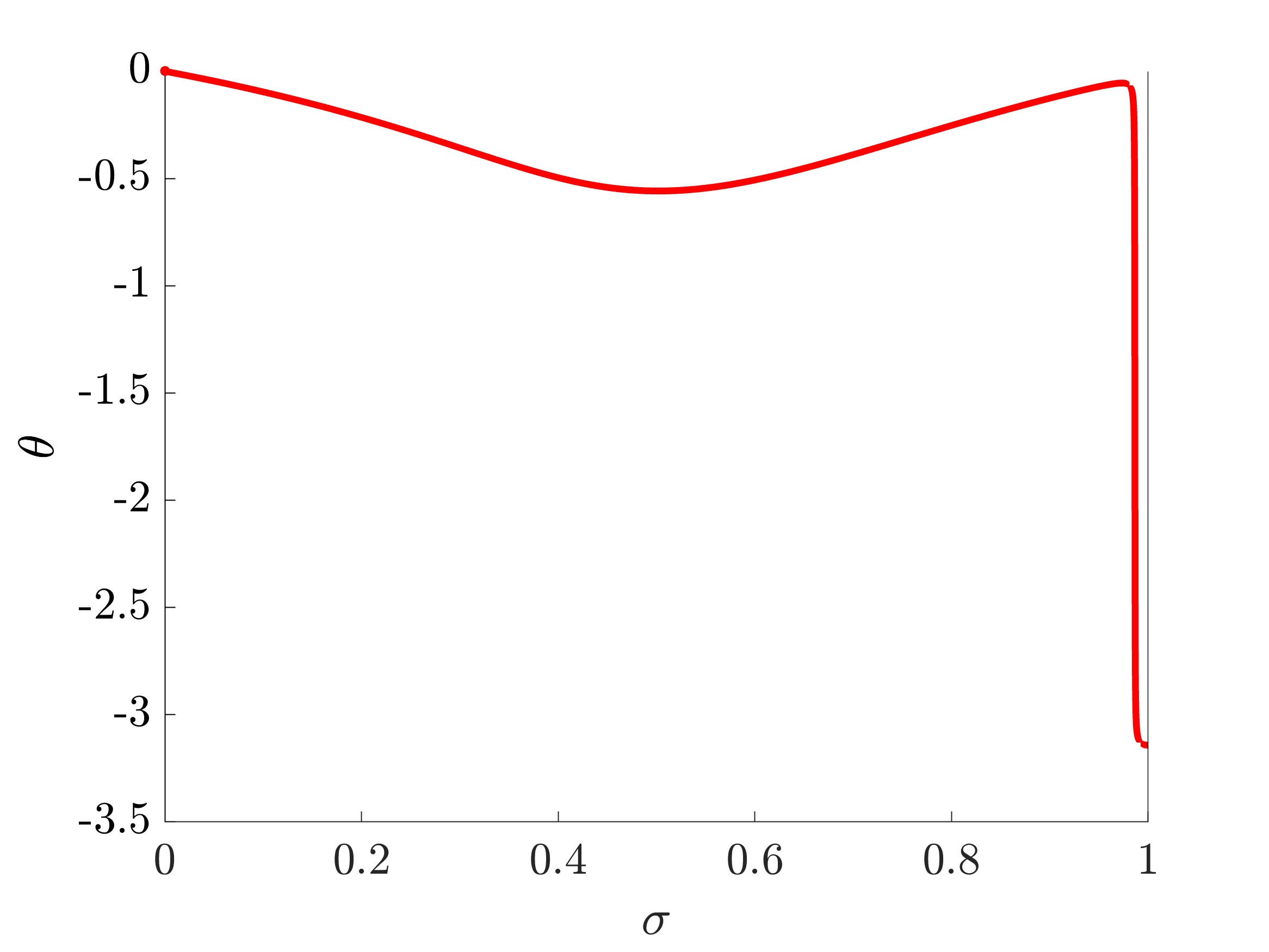}
}
\subfigure[\label{fig:Sigma-TauEc}]{
    \includegraphics[height=2.9cm]{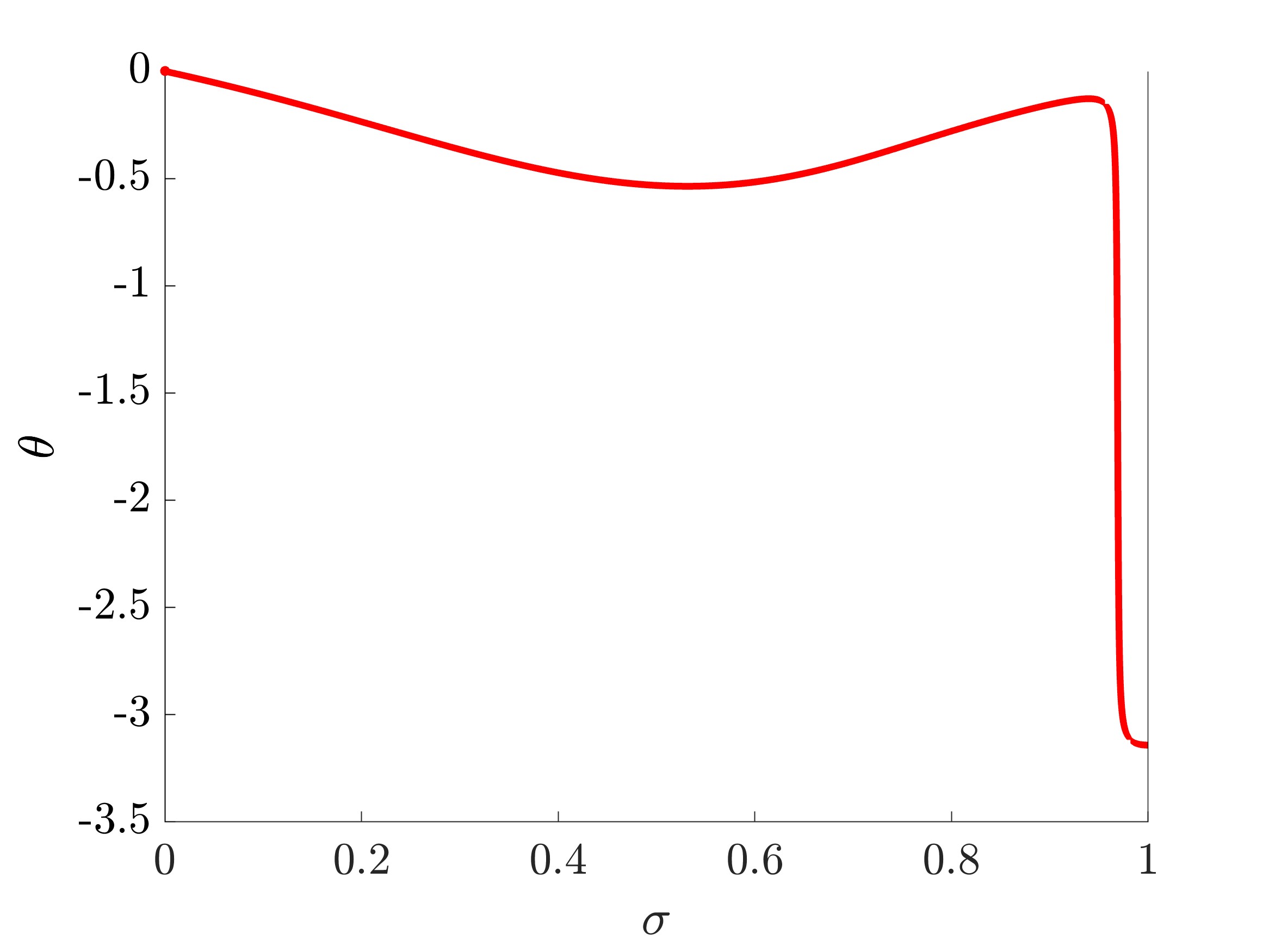}
}
\caption{Each panel shows the eigenvalue distribution computed using \eqref{eq:inner_model_evalue} under different scenarios: (a) Case 1, (b) Case 2, and (c) Case 3. All configurations converge to the same eigenvalue multiplicity, in that $m(V_0)=1$.}
\label{fig:mV0}
\end{center}
\end{figure}

\begin{figure}[tbp]
\begin{center}
\subfigure[\label{fig:Sigma-TauVa}]{
    \includegraphics[height=2.9cm]{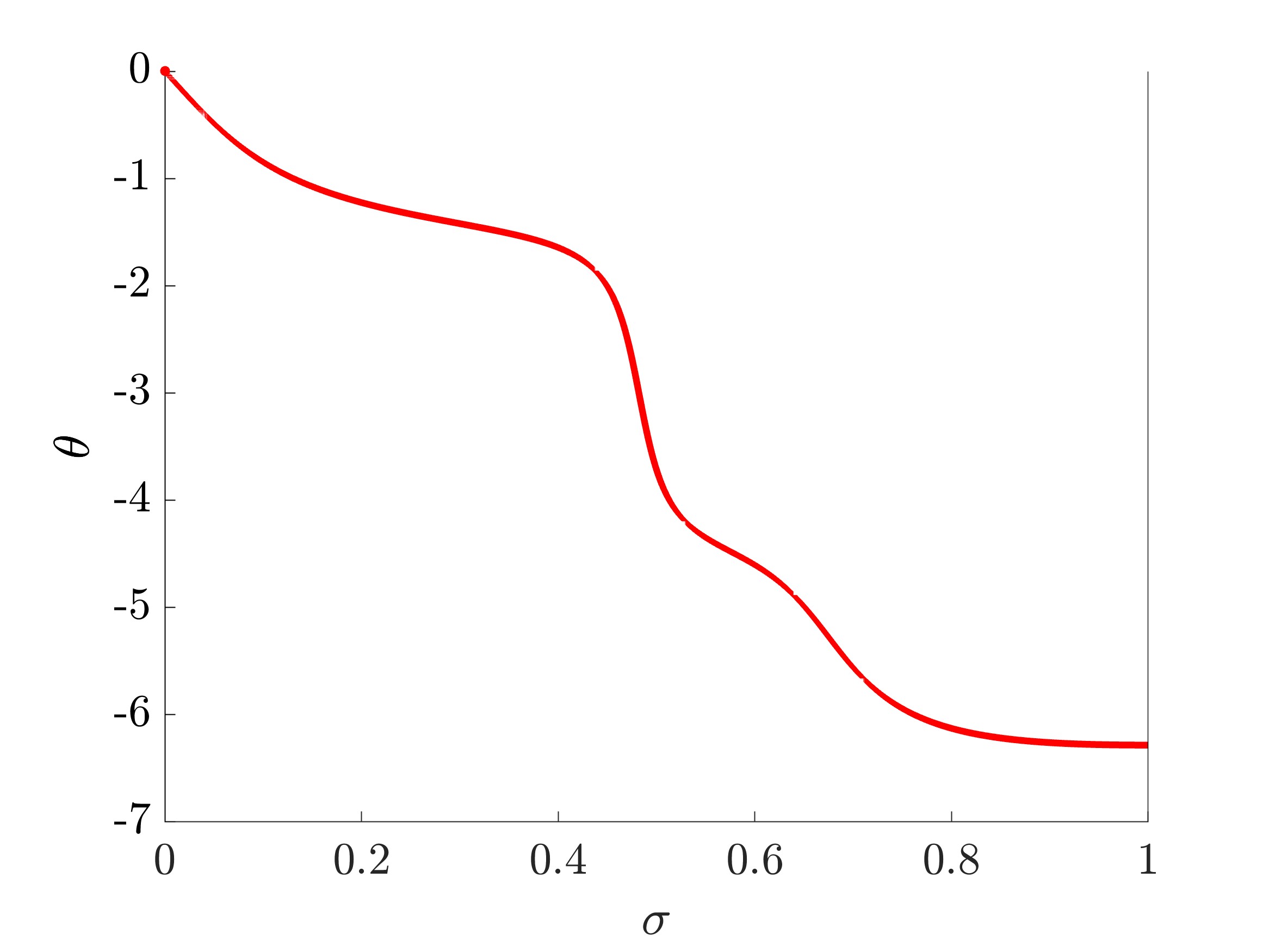}
}
\subfigure[\label{fig:Sigma-TauVb}]{
    \includegraphics[height=2.9cm]{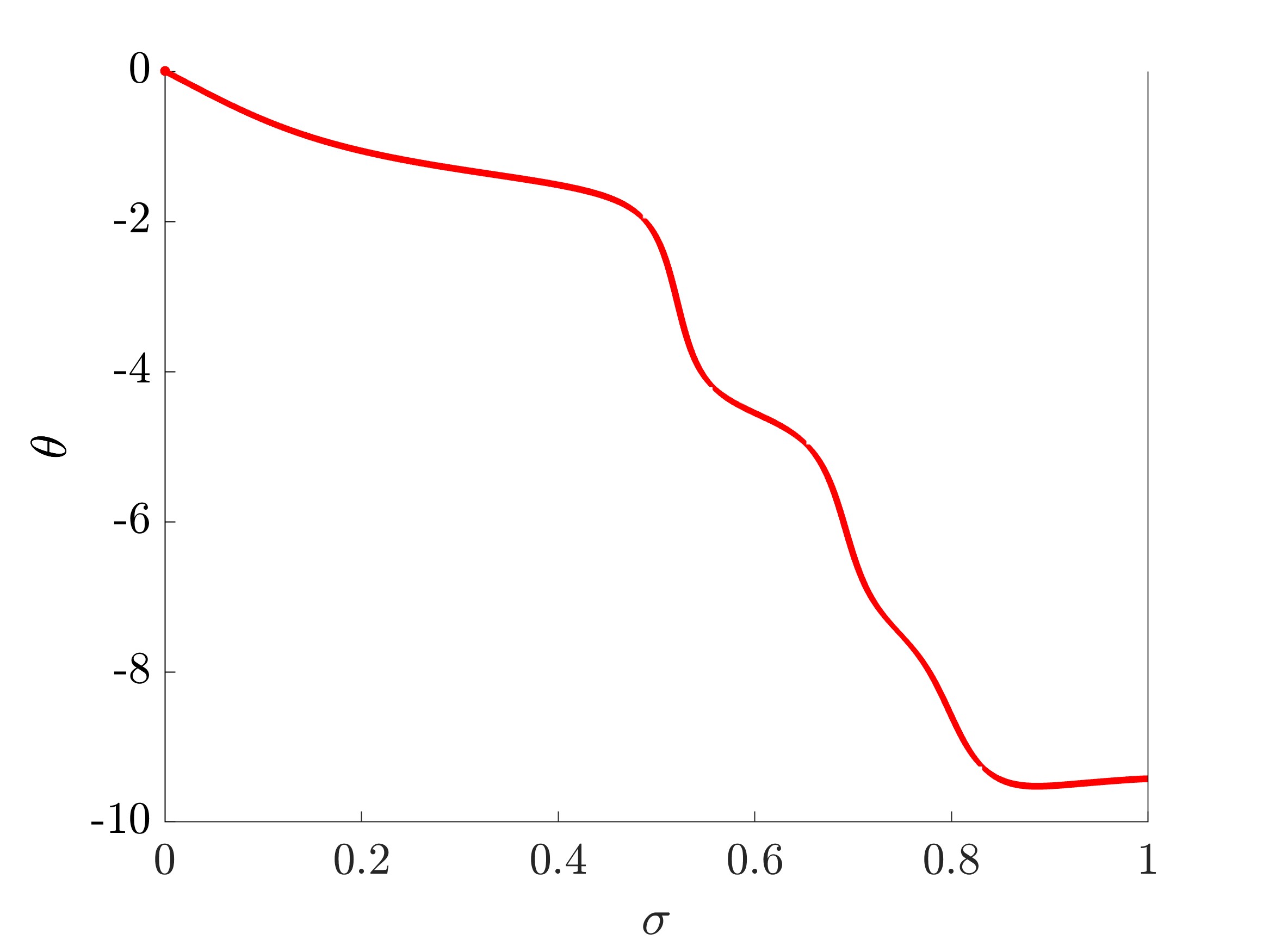}
}
\subfigure[\label{fig:Sigma-TauVc}]{
    \includegraphics[height=2.9cm]{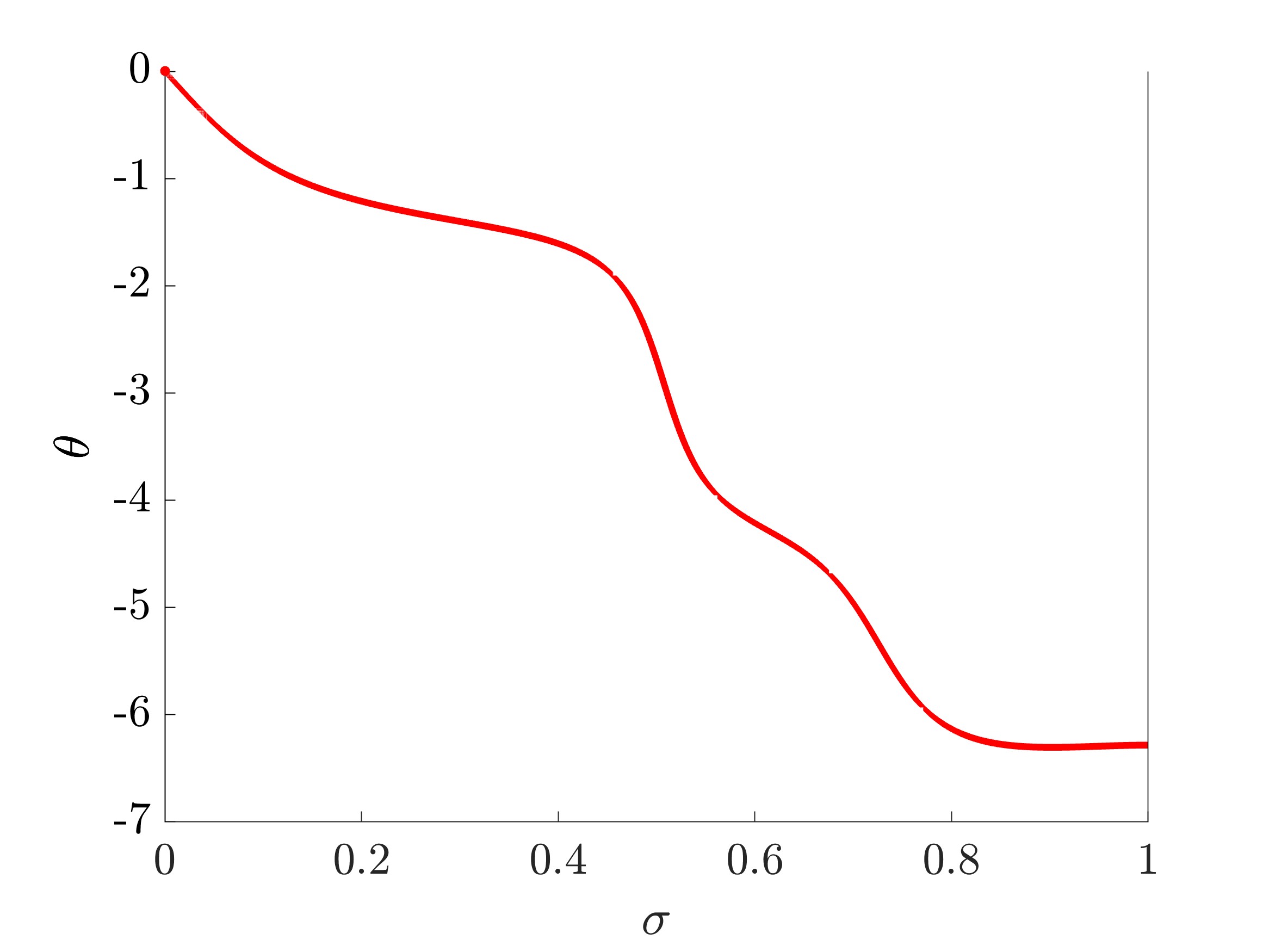}
}
\caption{Each panel shows the eigenvalue distribution computed using the $V_1$ model problem under different scenarios: (a) Case 1: $V_1=-30e^{-x^2}$, (b) Case 2: $V_1=-20\textrm{sech}(x)$, and (c) Case 3: $V_1 = \frac{-30}{(1+x^2)^2}$. For scenario 2, $m(V_1)=3$. For scenarios 1 and 3, $m(V_1)=2$.}
\label{fig:mV1}
\end{center}
\end{figure}

\subsection{Verification of spectral assumptions}
In all three scenarios, the operator $\Delta - V_0- V_{1,\varepsilon}$ satisfies the spectral assumptions required in Theorem~\ref{thm:1}. Specifically:
\begin{itemize}
    \item The essential spectrum of $\Delta - V_0$ remains $(-\infty, 0]$ by our assumptions on $V_0$ using Weyl's Theorem.
    \item Each $V_{1,\varepsilon}$ is a relatively compact perturbation preserving the essential spectrum and modifying the point spectrum by a finite number of eigenvalues.
    \item The eigenvalue count $m(V_1)$ 
    is well-defined and finite for each case.
\end{itemize}
Numerical results show at most three jumps in the angle variable $\psi$ from \eqref{eq:Lp6} in scenarios 1 and 3, and at most four jumps in scenario 2, corresponding to topological phase shifts in $\theta$ from \eqref{eq:Lp3}. The eigenvalue analysis in Figure \ref{fig:mV0} demonstrates that $m(V_0) = 1$ consistently across all scenarios. The eigenvalue count for the $V_1$ system for each scenario is illustrated in Figure \ref{fig:mV1}. We note that for scenario 2, $m(V_1)=3$, while in scenarios 1 and 3, $m(V_1)=2$. Incorporating this information into the topological relationship $m(W_{\varepsilon}) = m(V_0) + m(V_1)$, we obtain $m(W_{\varepsilon}) = 3$ for scenarios 1 and 3, and $m(W_{\varepsilon}) = 4$ for scenario 2. This observation is consistent for both the $\sigma$ and $\tau$ systems.

\begin{remark}\label{rem:transient}
The $\tau$ system trajectories in the right panels of Figures~\ref{fig:Sigma-Tau1},~\ref{fig:Sigma-Tau2} and ~\ref{fig:Sigma-Tau3} exhibit a characteristic initial transient that requires careful interpretation. Crucially, these transients contribute to the topological count given by $m(W_{\varepsilon})$, which essentially measures the net phase winding of the solution. The rapid initial transition represents genuine phase accumulation that must be included in this count.
\end{remark}

\section{Conclusion and discussion}\label{sec:conclusion}
In this paper, we have presented a rigorous analysis of the operator $\Delta - W_{\varepsilon}$ with scale-separated potential $W_{\varepsilon} = V_0 + V_{1,\varepsilon}$, where both components are smooth, radially symmetric, and satisfy assumptions (A1) and (A2). Our main focus has been on understanding the spectral properties, particularly the existence and counting of eigenvalues of different orders. More specifically, we have established that for sufficiently small $\varepsilon$, the number of eigenvalues of $\Delta - W_{\varepsilon}$ of order $\varepsilon^2$ is exactly equal to $m(V_1)$, where $m(V_1)$ counts the eigenvalues of $\Delta - V_{1,\varepsilon}$. This result provides deep insight into how scale separation affects the spectral properties of such operators and demonstrates a principle of spectral additivity for scale-separated potentials.

Our analysis employed a careful decomposition of the problem into two distinct spatial scale regimes, allowing us to track eigenvalues through a combination of Sturm-Liouville theory and perturbation methods. This approach demonstrates broad applicability to spectral problems involving multiple characteristic length scales.

In addition, we provided a complete characterization of the discrete spectrum by showing that for sufficiently small $\varepsilon$, the total number of eigenvalues of $\sigma\left(\Delta - W_{\varepsilon}\right)$ in the interval $(0, \infty)$ is precisely the sum of the contributions from each potential component: $m(V_1)$ from the scaled potential $V_{1,\varepsilon}$ and $m(V_0)$ from the fixed potential $V_0$.

\subsection{Future work}

Several promising avenues present themselves:

\noindent\textit{Regularity Requirements.} The smoothness assumptions on $V_0$ and $V_1$ could  
{\ef certainly} be weakened. Recent work by Frank-Simon~\cite{FrankSimon2015} suggests that finite regularity might suffice. For potentials in $L^p$ spaces, the results of Frank-Laptev-Safronov~\cite{FrankLaptevSafronov2016} provide guidance on possible extensions.

\noindent\textit{Non-radial Potentials.} Perhaps the most challenging extension would be to non-radially symmetric potentials. While some results exist for such potentials in the single-scale case, as shown by Agmon~\cite{Agmon1982}, the interaction of different scales without radial symmetry introduces substantial new difficulties. Recent techniques developed by Denisov~\cite{Denisov2004} for handling non-radial perturbations might prove useful in this context.

The techniques developed in this paper, particularly our scale separation analysis, should provide valuable insights for these extensions. Each relaxation of assumptions will require new mathematical tools, likely combining methods from spectral theory, harmonic analysis, and dynamical systems.

A key challenge will be understanding how the eigenvalue count in the spectral gap persists under these various generalizations. The work of Simon~\cite{Simon1982} on eigenfunction decay and the general framework of Reed-Simon~\cite{reed1978iv}, Chapter XIII for essential spectra will likely play crucial roles in these investigations.

\section*{Acknowledgments} JLM thanks Kenji Nakanishi for framing the original problem and acknowledges the thesis work of his student Dmitro Golovanich for laying the foundation of finding its solution.  EF was supported by the National Science Foundation under Grant No. DMS-2137947.  JLM gratefully acknowledges support from NSF grant DMS-2307384. The research of CKRTJ was supported by the US Office of Naval Research under grant N00014-24-1-2198.  AI tools (ChatGPT 5.6) were used in the revision for reference checking and transcribing notes. \correction{Thanks}{We thank the anonymous referees for carefully reading the result and making many important suggestions to improve and correct various parts of the manuscript.}

\bibliographystyle{alpha}
\bibliography{references_scalar}

@article{FrankLaptevSafronov2016,
    author = {Frank, R. L. and Laptev, A. and Safronov, O.},
    title = {On the number of eigenvalues of {Schr\"odinger} operators with complex potentials},
    journal = {J. Math. Phys.},
    volume = {57},
    year = {2016},
    pages = {122101}
}

@article{Denisov2004,
    author = {Denisov, S. A.},
    title = {On the absolutely continuous spectrum of {Dirac} operator},
    journal = {Comm. Partial Differential Equations},
    volume = {29},
    year = {2004},
    pages = {1403--1428}
}

@article{FrankSimon2015,
    author = {Frank, R. L. and Simon, B.},
    title = {Eigenvalue bounds for {Schr\"odinger} operators with complex potentials. {II}},
    journal = {J. Spectr. Theory},
    year = {2015}
}

@book{Kato1995,
    author = {Kato, T.},
    title = {Perturbation Theory for Linear Operators},
    publisher = {Springer-Verlag},
    year = {1995}
}

@book{Agmon1982,
    author = {Agmon, S.},
    title = {Lectures on Exponential Decay of Solutions of Second-Order Elliptic Equations},
    publisher = {Princeton University Press},
    year = {1982}
}

@article{Simon1982,
    author = {Simon, B.},
    title = {Schr\"odinger semigroups},
    journal = {Bull. Amer. Math. Soc.},
    volume = {7},
    year = {1982},
    pages = {447--526}
}

@article{Cwikel1977,
  title={Weak type estimates for singular values and the number of bound states of {S}chr{\"o}dinger operators},
  author={Cwikel, Michael},
  journal={The Annals of Mathematics},
  volume={106},
  number={1},
  pages={93--100},
  year={1977},
  publisher={JSTOR}
}

@article{Lieb1976,
  title={Bounds on the eigenvalues of the {L}aplace and {S}chr{\"o}dinger operators},
  author={Lieb, Elliott},
  journal={Bulletin of the American Mathematical Society},
  volume={82},
  number={5},
  pages={751--754},
  year={1976}
}

@article{Rozenbljum1976,
  title={Distribution of the discrete spectrum of singular differential operators},
  author={Rozenbljum, Grigorii Vladimirovich},
  journal={Izvestiya Vysshikh Uchebnykh Zavedenii Matematika},
  volume={202},
  number={1},
  pages={75--86},
  year={1976}
}

@article{lieb2001inequalities,
  title={Inequalities for the moments of the eigenvalues of the {S}chr{\"o}dinger {H}amiltonian and their relation to {S}obolev inequalities},
  author={Lieb, Elliott H and Thirring, Walter E},
  journal={The Stability of Matter: From Atoms to Stars: Selecta of Elliott H. Lieb},
  pages={205--239},
  year={2001},
  publisher={Springer}
}

@inproceedings{berthier1982point,
  title={On the point spectrum of {S}chr{\"o}dinger operators},
  author={Berthier, Anne},
  booktitle={Annales scientifiques de l'{\'E}cole Normale Sup{\'e}rieure},
  volume={15},
  pages={1--15},
  year={1982}
}

@article{sandstede2004evans,
  title={Evans function and blow-up methods in critical eigenvalue problems},
  author={Sandstede, Bj{\"o}rn and Scheel, Arnd},
  journal={Discrete and Continuous Dynamical Systems},
  volume={10},
  number={2004},
  year={2004}
}

@article{wieczorek2021compactification,
  title={Compactification for asymptotically autonomous dynamical systems: theory, applications and invariant manifolds},
  author={Wieczorek, Sebastian and Xie, Chun and Jones, Chris KRT},
  journal={Nonlinearity},
  volume={34},
  number={5},
  pages={2970},
  year={2021},
  publisher={IOP Publishing}
}

@book{reed1978iv,
  title={Methods of Modern Mathematical Physics IV: Analysis of Operators},
  author={Reed, Michael and Simon, Barry},
  volume={4},
  year={1978},
  publisher={Elsevier}
}

@book{reed1978iii,
  title={Methods of Modern Mathematical Physics III: Scattering Theory},
  author={Reed, Michael and Simon, Barry},
  volume={3},
  year={1978},
  publisher={Elsevier}
}

@article{nakanishi2017global1,
  title={Global dynamics above the first excited energy for the nonlinear {S}chr{\"o}dinger equation with a potential},
  author={Nakanishi, Kenji},
  journal={Communications in Mathematical Physics},
  volume={354},
  number={1},
  pages={161--212},
  year={2017},
  publisher={Springer}
}

@article{nakanishi2017global,
  title={Global dynamics below excited solitons for the nonlinear {S}chr{\"o}dinger equation with a potential},
  author={Nakanishi, Kenji},
  journal={Journal of the Mathematical Society of Japan},
  volume={69},
  number={4},
  pages={1353--1401},
  year={2017},
  publisher={The Mathematical Society of Japan}
}

@article{hundertmark2023cwikel,
  title={Cwikel's bound reloaded},
  author={Hundertmark, Dirk and Kunstmann, Peer and Ried, Tobias and Vugalter, Semjon},
  journal={Inventiones mathematicae},
  volume={231},
  number={1},
  pages={111--167},
  year={2023},
  publisher={Springer}
}

@phdthesis{golovanich2021potential,
  title={Potential Perturbations of the 3D Nonlinear {S}chr\"odinger Equation},
  author={Golovanich, Dmitro},
  year={2021},
  school={The University of North Carolina at Chapel Hill}
}

@article{KriegerSchlag2006,
  title={Stable manifolds for all monic supercritical focusing nonlinear {S}chr{\"o}dinger equations in one dimension},
  author={Krieger, Joachim and Schlag, Wilhelm},
  journal={Journal of the American Mathematical Society},
  volume={19},
  number={4},
  pages={815--920},
  year={2006}
}

@book{NakanishiSchlag2011,
  title={Invariant manifolds and dispersive {H}amiltonian evolution equations},
  author={Nakanishi, Kenji and Schlag, Wilhelm},
  year={2011},
  publisher={European Mathematical Society}
}

@article{bargmann1952number,
  title={On the number of bound states in a central field of force},
  author={Bargmann, Valentine},
  journal={Proceedings of the National Academy of Sciences},
  volume={38},
  number={11},
  pages={961--966},
  year={1952}
}

@book{perko2013differential,
  title={Differential equations and dynamical systems},
  author={Perko, Lawrence},
  year={2013},
  publisher={Springer Science \& Business Media}
}

\end{document}